\documentclass[11pt]{article}
\usepackage{color}
\usepackage{epic}
\usepackage{amsthm,amsmath,amsfonts,mathtools,mathrsfs,amssymb,latexsym,bm,soul}
\usepackage{graphicx}
\usepackage{setspace}
\usepackage{comment,  verbatim} 
\usepackage{enumerate,enumitem}
\usepackage{graphicx}
\usepackage{float}
\usepackage{setspace}
\usepackage{hyperref}
\usepackage{subcaption}
\usepackage[table]{xcolor}
\usepackage{setspace}
\usepackage[margin=1in]{geometry}
\usepackage[affil-it]{authblk}
\usepackage{color}
\usepackage{amsmath}
\usepackage{graphicx,psfrag,epsf}
\usepackage{enumerate}
\usepackage[comma,numbers]{natbib}
\usepackage{algorithm,algpseudocode}

\usepackage[normalem]{ulem}

\input{./Definitions}

\allowdisplaybreaks
\def\T{{ \mathrm{\scriptscriptstyle T} }}
\def\diag{{ \mathrm{diag} }}

\theoremstyle{plain}

\title{Gaussian Process Regression and Classification using International Classification of Disease Codes as Covariates}

\author[1]{Sanvesh Srivastava \thanks{\url{sanvesh-srivastava@uiowa.edu}}}
\author[1]{Zongyi Xu \thanks{\url{zongyi-xu@uiowa.edu}}}
\author[2]{Yunyi Li \thanks{\url{Yunyi.Li@mccombs.utexas.edu}}}
\author[3]{W. Nick Street \thanks{\url{nick-street@uiowa.edu}}}
\author[4]{Stephanie Gilbertson-White \thanks{\url{stephanie-gilbertson-white@uiowa.edu}}}

\affil[1]{Department of Statistics and Actuarial Science, The University of Iowa}
\affil[2]{McCombs School of Business, University of Texas Austin}
\affil[3]{Department of Business Analytics, The University of Iowa}
\affil[4]{College of Nursing, The University of Iowa}

\date{\today}

\begin{document}

\maketitle
\begin{abstract}
  International Classification of Disease (ICD) codes are widely used for encoding diagnoses in electronic health records (EHR). Automated methods have been developed over the years for predicting biomedical responses using EHR that borrow information among diagnostically similar patients. Relatively less attention has been paid on developing patient similarity measures that model the structure of ICD codes and the presence of multiple chronic conditions, where a chronic condition is defined as a set of ICD codes. Motivated by this problem, we first develop a type of string kernel function for measuring similarity between a pair of subsets of ICD codes, which uses the definition of chronic conditions. Second, we extend this similarity measure to define a family of covariance functions on subsets of ICD codes. Using this family, we develop Gaussian process (GP) priors for Bayesian nonparametric regression and classification using diagnoses and other demographic information as covariates. Markov chain Monte Carlo (MCMC) algorithms are used for posterior inference and predictions. The proposed methods are free of any tuning parameters and are well-suited for automated prediction of continuous and categorical biomedical responses that depend on chronic conditions. We evaluate the practical performance of our method on EHR data collected from 1660 patients  at the University of Iowa Hospitals and Clinics (UIHC) with six different primary cancer sites. Our method has better sensitivity and specificity than its competitors in classifying different primary cancer sites and estimates the marginal associations between chronic conditions and primary cancer sites.  
\end{abstract}


\section{Introduction}

\label{s:intro}

EHR data are useful for organizing information about clinical care and developing tools for personalized patient care. They are also used for answering interesting and important research questions that cannot be practically addressed through traditional prospective research. We develop a Bayesian model based on GP priors for predicting biomedical responses depending on the presence of a set of chronic conditions. We assume that the EHR encode diagnoses using the 10th edition of ICD billing code (hereafter called ICD code), so a chronic condition is defined using one or more specific codes from a predefined set of ICD codes \citep{Caletal17}. An ICD code is a three to seven character long string, where the first three characters define the diagnosis category, the next three define the etiology, anatomic site, and severity, and the final character identifies encounters. We develop a family of kernel functions on a pair of ICD code subsets that is based on the definition of ICD codes and chronic conditions. Using any such kernel as a covariance function for a GP prior, we perform Bayesian nonparametric regression and classification using patients' chronic conditions and demographic information as covariates.

Encoding diagnoses using ICD coding system has become the most widely adopted medical diagnostic classification system since it was endorsed by the World Health Assembly in 1990. Automated extraction of data from the EHR using ICD codes allows us to obtain large amounts of precise medical information about chronic conditions. This is a cheap and efficient alternative to the significantly more expensive and slower manual chart review by a healthcare provider. Unfortunately, ICD codes are prone to encoding errors, so care must be exercised in fitting a model with ICD codes as covariates \citep{Honetal19}. This has motivated a significant interest in automating the ICD encoding process to eliminate any sources of errors \citep{Xuetal19}. In a properly curated EHR, models that use ICD codes as covariates in predicting disease risks and related medical conditions are known to be more accurate than those which do not use ICD codes \citep{Horetal17}. Building on these ideas, we achieve an additional but significant improvement in predictive accuracy by exploiting the grouping of diagnoses into chronic conditions and the hierarchical structure of ICD codes. Our Bayesian methods require minimal tuning and provide results with uncertainty estimates, so they can be easily used by a healthcare provider for assisting with patient care.

EHR data have been analyzed using Bayesian nonparametric methods, but the literature on using ICD codes as covariates is sparsely populated. \citet{Henetal16} have developed a {deep} hierarchical factor model for identifying latent subgroups of patients. They summarize the EHR data in terms of counts of different medications usage, laboratory tests, and diagnoses specified using ICD codes. If the EHR data only contains information about the presence or absence of disease, then Bayesian nonparametric extensions of low-rank matrix factorization perform better in discovering latent disease groups \citep{Nietal19}. Recently, deep learning has been widely used for automated regression and classification using EHR data, but the major focus remains on mining clinical notes and not on using the information encoded in ICD codes or on providing uncertainty estimates \citep{Rajetal18}. A major limitation of all these approaches is that they ignore the additional structure provided by the membership of diagnoses in different chronic conditions.

Addressing this gap, we first develop a similarity measure on diagnoses encoded as ICD codes that accounts for the grouping of diagnoses in chronic conditions and the hierarchical structure of ICD codes. A widespread practice is to define  the similarity of two patients to be proportional to the number of common ICD codes in their diagnoses or Jaccard index. Our similarity measure is superior to this practice in that its value could be positive for a pair of non-identical ICD codes, and its magnitude depends on the degree of overlap between the pair. In fact, our similarity measure belongs to the class of string kernels and extends a restricted version of the \emph{boundrange} kernel, which has been widely used in text mining and information retrieval \citep{Lodetal02,ShaCri04}. The similarity between a pair of ICD codes is defined as the total number of common sub-strings between them, where a sub-string always begins at the first character of an ICD code. This similarity measure is extended to a pair of subsets as a weighted average of similarities of all ICD code pairs formed using the two subsets.

Our second major contribution is to use the similarity between patients for nonparametric regression and classification using GP priors. This is done in two stages. First, using the similarity measure, we define the equivalents of polynomial, exponential, and squared exponential kernels on subsets of ICD codes, which represent chronic conditions. Second, we use these kernels as covariance functions for defining GP priors indexed by subsets of ICD codes. These GP priors are used for fitting Bayesian nonparametric regression and classification models that are 
tuned for EHR data analysi \citep{RasWil06}. We develop MCMC algorithms for automated fitting of these models that provide theoretically guaranteed estimation of posterior uncertainty \citep{GhoVan17}. Minimal tuning requirements makes our method extremely attractive for routine applications in analyzing EHR data and for assisting healthcare providers in predicting risks related to chronic conditions. We verify our claims through empirical studies and show that the proposed method provides a more detailed quantification of dependence of primary cancer sites on chronic conditions. 

\section{Motivating EHR Data}
\label{s:mot-data}

Multimorbidity is the simultaneous occurrence of more than one chronic condition. Developing operational measures of multimorbidity is an active research area due to its importance in improving patient care. \citet{Caletal17} have operationalized the concept of multimorbidity into a list of 918 ICD codes that map on to a list of 60 chronic conditions. This serves as a motivation for developing a classifier that accounts for the grouping of diagnoses into chronic conditions. We show that this yields better accuracy in predicting primary cancer sites in UIHC EHR data.

The data for this analysis came from the EHR at UIHC relating to the set of patients who received cancer treatment beginning in 2017. The cohort included 18 years or older patients who were diagnosed with malignant solid neoplasm of any type or site. The data included diagnoses of the patients encoded as ICD codes and the race and marital status of every patient. We filtered patients who had at least one of the 58 non-cancerous chronic conditions, resulting in a sample size of 1660. The main question here is to identify chronic conditions or diagnoses that are predictive of a primary cancer site. Unfortunately, the Bayesian toolkit provides limited options for answering such biomedical questions. Using the UIHC EHR data as a representative example, we develop GP-based regression and classification models that are tuned for solving such biomedical problems, where the subsets of ICD codes are covariates.

The main goal is to predict marginal associations between chronic conditions and cancer primary sites and a minor goal is to predict the primary cancer site using the ICD codes, which are structured strings; for example, J301 encodes allergy and D50, D51 denote different kinds of anaemia. Clearly, using dummy variables for representing these ICD codes is inappropriate because D50, D51 denote the same disease but are assigned different dummy variables. Motivated from the text mining literature, a useful similarity measure for strings is defined by counting the number of matching substrings of different lengths. Accounting for the hierarchical structure of ICD codes, if we modify this definition by enforcing the substrings to always start at the beginning, then the similarities of the three pairs (J301, D50), (J301, D51), and (D50, D51) are 0, 0, 2, respectively, which are more realistic. In the next section, we develop this idea further to account for the extra structure provided by the 58 chronic conditions while down-weighting the contributions of commonly occurring diagnoses.

\section{Kernel Functions for Diagnoses}

\subsection{Kernel Function on ICD Codes}

Consider the structure of ICD codes. Let $t$ be an ICD code, $t^*$ be the longest ICD code, $| t |$ denote the number of characters in $t$, and $t_{1:i}$ be the contiguous substring of length $i$ that starts at the first character of $t$. The first character of $t$ is  $t_{1:1}$ and $t_{1:l} = t$ if $l > |t|$. Let $\Tcal$ be the set of all ICD codes, $\Acal$ and $\Dcal$ be the sets of letters and numbers that are used to represent an ICD code, and $\Bcal = \Acal \left( \Acal \cup \Dcal \right)^{|t^*| - 1}$ is set of all alphanumeric substrings that have an letter as their first character and have lengths $1, 2, \ldots,$ or $|t^*|$. Then, $\Bcal$ contains all contiguous substrings of ICD codes in $\Tcal$ that start at the first character in an ICD code and $|\Tcal| < |\Bcal| = |\Acal| \left( |\Acal| + |\Dcal| \right)^{|t^*| - 1}$.  

The kernel function on two ICD codes is defined as a Euclidean dot product between their feature maps. Let $\psib(t)$ be the feature map of an ICD code $t \in \Tcal$ and $\Fcal = \{\psib(t): t \in \Tcal\}$ be the feature space. Then, for every $t \in \Tcal$, $\psib(t) \in \{0, 1\}^{|t^*||\Bcal|}$ and is defined as 
\begin{align}
  \label{eq:1}
  \psib: \Tcal \mapsto \Fcal, \quad \psib(t) =  \{\psib_b(t)\}_{b \in \Bcal}, \quad \psib_b(t) = \{1(b = t_{1:1}), \ldots, 1(b = t_{1:|t^*|})\}^\T, 
\end{align}
where $1(\cdot)$ is 1 if $\cdot$ is true and 0 otherwise. The kernel function on $\Tcal \times \Tcal$ is defined using $\psib(\cdot)$ and known $|t^*| \times |t^*|$ symmetric positive semi-definite matrices $(\Lambdab_b)_{b \in \Bcal}$ as 
\begin{align}
  \label{eq:e1}
  \kappa_0(t, t' \mid \Lambdab) = \underset{b \in \Bcal} {\sum} \psib^\T_b(t) \Lambdab_b \psib_b(t') \equiv \psib^\T(t) \Lambdab \psib(t'), \; t, t' \in \Tcal, 
\end{align}
where $\Lambdab = \diag\{\Lambdab_b: b \in \Bcal\}$ is a block diagonal matrix. Two useful special cases of \eqref{eq:e1} are obtained by setting $\Lambdab_b$ to be a diagonal matrix  in \eqref{eq:e1} such that (i) $\Lambdab_b = \lambda_b I$ and (ii) $(\Lambdab_b)_{aa} = \lambda^a$ for every $b \in \Bcal$ and $a = 1, \ldots, |t^*|$, where $I$ is an identity matrix, $\lambda_b > 0$  for every $b \in \Bcal$, and $\lambda \in (0, 1)$. We prove in the next section that $\kappa_0(t, t' \mid \Lambdab)$ in \eqref{eq:e1} is a valid kernel function on $\Tcal \times \Tcal$.

The kernel in \eqref{eq:e1} belongs to the class of string kernels that are widely used in text mining and information retrieval \citep{ShaCri04}. Two popular string kernels are spectrum and boundrange kernels \citep{Lodetal02}. The $k$-spectrum kernel defines the similarity to be the number of matching substrings between two strings of length $k$ exactly. The $k$-boundrange kernel instead defines the similarity to be the number of matching substrings between two strings of length less than or equal to $k$. The substrings can be non-contiguous in both kernels and the substring matches can be weighted by a factor that decays exponentially with the substring length. The kernel in \eqref{eq:e1} is restricted version of $|t^*|$-boundrange kernel where the substrings always start at the beginning of an ICD code.  

\subsection{Kernel Function on Subsets of ICD Codes}

We now define the kernel function on subsets of ICD codes that represent diagnoses. Let $2^{\Tcal}$ be the power set of $\Tcal$, $\tb = \{t_1, \ldots, t_r \}$ be an element of $2^{\Tcal}$, $\Psib(\tb)$ be the feature map of $\tb$, and $\Fcal_{\Psib} = \{\Psib(\tb): \tb \in 2^{\Tcal}\}$ be the feature space of subsets of ICD codes. The feature vector $\Psib(\tb)$ and the kernel function $\kappa_1$ on $2^{\Tcal} \times 2^{\Tcal}$ are defined using 
$\psib(\cdot)$  in \eqref{eq:1} and $\kappa_0(\cdot, \cdot \mid \Lambdab)$ in \eqref{eq:e1}, respectively, as
\begin{align}
  \label{eq:3}
  \Psib(\tb) = \sum_{{t \in \tb}} \psib(t) , \quad \kappa_1(\tb, \tb' \mid \Lambdab ) = \sum_{{t \in \tb}} \sum_{{t' \in \tb'}} \kappa_0(t, t' \mid \Lambdab) = \Psib^\T(\tb) \Lambdab \Psib(\tb') , 
\end{align}
for $\tb, \tb' \in 2^{\Tcal}$, where $\Psib(\tb)$ is a $(|t^*||\Bcal|)$-dimensional vector with entries in $\{0\} \cup \NN $. The kernel $\kappa_1$ is a \emph{derived subsets} kernel with the \emph{base} kernel $\kappa_0$; see Definition 9.41 in \citet[Page 317]{ShaCri04}. 

In the motivating EHR data, covariates include subsets of ICD codes that have an additional structure specified by their relation to a chronic condition. Let $C$ be the total number of chronic conditions and $\Tb_c \in 2^{\Tcal}$ be the subset of ICD codes that defines the chronic condition $c$ $(c=1, \ldots, C)$. A patient has chronic condition $c$ if the patient's diagnoses consists of an ICD code that belongs to $\Tb_c$. Accounting for the $C$ chronic conditions, represent a patient's diagnoses as the set $\Tb = \{\tb_1, \ldots, \tb_c, \ldots, \tb_C\} \in \left( 2^{\Tcal} \right)^C $, where $\tb_c \subseteq \Tb_c$. A kernel function on $\left( 2^{\Tcal} \right)^C \times \left( 2^{\Tcal} \right)^C$ that accounts for the structure imposed by $C$ chronic conditions is
\begin{align}
  \label{eq:4}
  \kappa_2(\Tb, \Tb' \mid \wb, \Lambdab) = \sum_{c=1}^C w_c \, \kappa_1(\tb_c, \tb'_c \mid \Lambdab), \quad \Tb, \Tb' \in \left( 2^{\Tcal} \right)^C,  \quad w_c > 0, 
\end{align}
where $\Lambdab$ is defined in \eqref{eq:e1}, $\wb = (w_1, \ldots, w_C)$, and $w_c$ indicates the importance of chronic condition $c$ in defining the similarity between two patients with diagnoses $\Tb$ and $\Tb'$, respectively. 
The kernel $\kappa_2$ assumes the effect of $C$  chronic conditions on the response is additive, but other valid methods for combining kernels $\kappa_1(\tb_1, \tb'_1 \mid \Lambdab), \ldots, \kappa_1(\tb_C, \tb'_C \mid \Lambdab)$ can be also used in \eqref{eq:4}. 

The kernel $\kappa_2$ in \eqref{eq:4} is affected by the number of elements in $\Tb$ and $\Tb'$. Let $\| \cdot \|$ be the Euclidean norm. The value of $\kappa_{1}$ in \eqref{eq:3} is impacted by $\tb$ in that $\| \Psib(\tb) \|$ increases with the cardinality of $\tb$. This implies that $\kappa_{2}$ in \eqref{eq:4} is large if the cardinality of $\Tb$ is large; therefore, the length of feature maps impacts the similarity measure instead of the chronic conditions. 
This undesirable effect is removed by normalizing the feature maps of $\Tb, \Tb'$ before computing the dot product, which results in the kernel
\begin{align}
  \label{eq:7}
  \kappa (\Tb, \Tb' \mid \wb, \Lambdab) &= \frac{\kappa_2 (\Tb, \Tb' \mid \wb, \Lambdab)}{\left\{ \kappa_2 (\Tb, \Tb \mid \wb, \Lambdab)\right\}^{1/2} \left\{ \kappa_2 (\Tb', \Tb' \mid \wb, \Lambdab) \right\}^{1/2}}.
\end{align}
All kernels in this paper are derived from $\kappa$ in \eqref{eq:7}, which depends on the kernels $\kappa_0$, $\kappa_1$, and $\kappa_2$ defined earlier.

The following theorem proves that $\kappa$ is a valid kernel function on  $(2^{\Tcal})^C \times (2^{\Tcal})^C$ and defines a metric and the equivalents of polynomial, exponential, and squared exponential (SE) kernels on $(2^{\Tcal})^C \times (2^{\Tcal})^C$, respectively. 
\begin{theorem}\label{thm1}
  Let $\{\Lambdab_b\}_{b \in \Bcal}$ and $\wb$ be defined as in \eqref{eq:e1} and \eqref{eq:4}, $\Tb, \Tb' \in (2^{\Tcal})^C$, $\sigma^2 > 0$, $\phi > 0 $, and
  $$d(\Tb, \Tb' \mid \wb, \Lambdab) = \left\{ \kappa(\Tb, \Tb \mid \wb, \Lambdab) + \kappa(\Tb', \Tb' \mid \wb, \Lambdab) - 2 \kappa(\Tb, \Tb' \mid \wb, \Lambdab) \right\}^{1/2}.$$
  Then, given $\Lambdab$ and $\wb$,
  \begin{enumerate}
  \item $\kappa(\cdot, \cdot \mid \wb, \Lambdab)$ and $d(\cdot, \cdot \mid \wb, \Lambdab)$ are valid kernel and distance functions on $(2^{\Tcal})^C \times (2^{\Tcal})^C$; 
  \item if $\Lambdab = \Rb^{\top} \Rb$, where $\Rb$ is an upper triangular matrix, then the feature map of $\kappa$ is 
    \begin{align}
      \label{eq:thm-feat}
      \overline \Psib_{\wb, \Rb} (\Tb) = \frac{\Psib_{\wb, \Rb} (\Tb)}{\| \Psib_{\wb, \Rb} (\Tb) \|}, \quad \Psib_{\wb, \Rb} (\Tb) = \{w_1^{1/2} \Rb \Psib(\tb_1), \ldots, w_C^{1/2} \Rb \Psib(\tb_C)\}^\T,
    \end{align}
    where $\Psib_{\wb, \Rb} (\Tb)$ is a $(C|t^*||\Bcal|)$-dimensional vector; and 
  \item the equivalents of polynomial kernel of order $s$ and  $\gamma$-exponential kernel on $(2^{\Tcal})^C \times (2^{\Tcal})^C$, respectively, are defined as
  \end{enumerate}  
  \begin{align}
    \label{eq:thm-kern}
    \kappa_s(\Tb, \Tb' \mid \sigma^2, \wb, \Lambdab) &=  \sigma^2\left\{1 + \kappa(\Tb, \Tb' \mid \wb, \Lambdab) \right\}^s, \quad  s \in \NN,
                                                      \nonumber \\
    \kappa_\gamma^e(\Tb, \Tb'  \mid \sigma^2, \phi, \wb, \Lambdab) &= \sigma^2 e^{- \phi \, \{d(\Tb, \Tb' \mid  \wb, \Lambdab)\}^\gamma}, \quad \gamma \in [1, 2],
  \end{align}
  where $\sigma^2$ and $\phi$  play the role of variance and inverse length-scale parameters, respectively, and $\kappa_\gamma^e$ reduces to the exponential and SE kernels when $\gamma $ equals 1 and 2, respectively. 
\end{theorem}
The proof of this theorem is in the supplementary material along with other proofs. The polynomial kernel has finite dimensional feature. If $s > 1$, then the polynomial kernel $\kappa_s$ is obtained by a non-linear embedding of the kernel $\kappa$. The dimension of the feature space for $\kappa_s$ is  ${C|t^*||\Bcal| + s \choose s}$  \citep[Proposition 9.2,][]{ShaCri04}. The feature space of $\kappa_{\gamma}^e$ lies in a Hilbert spaces of functions, where the ``smoothness'' of a function depends on the chosen kernel; see Section \ref{sub-sec:theory} for the precise definitions and greater details. We use the SE kernel in our experiments due to its popularity.

\begin{figure}
  
  \begin{subfigure}{.5\textwidth}
    \centering
    \includegraphics[width=.8\linewidth]{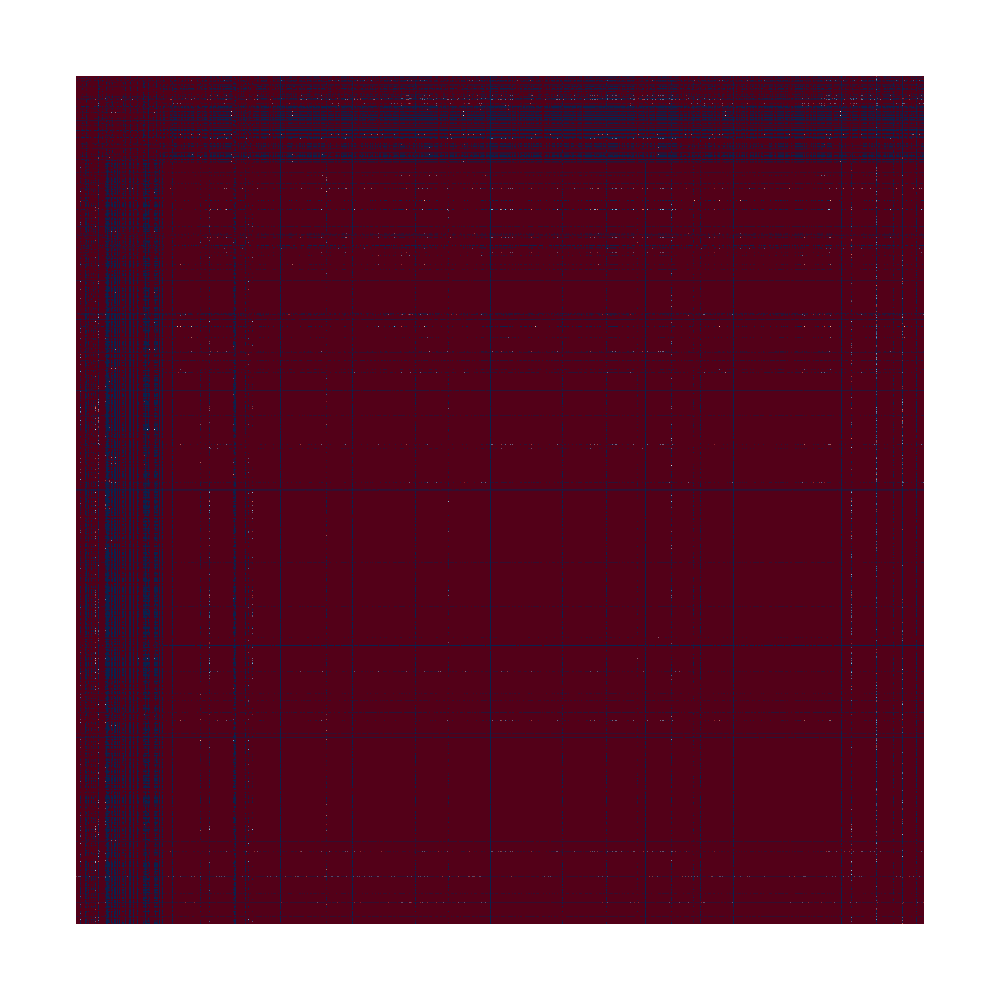}
    \caption{Kernel matrix of $\kappa_2$.}
    \label{fig:sfig1}
  \end{subfigure}%
  \begin{subfigure}{.5\textwidth}
    \centering
    \includegraphics[width=.8\linewidth]{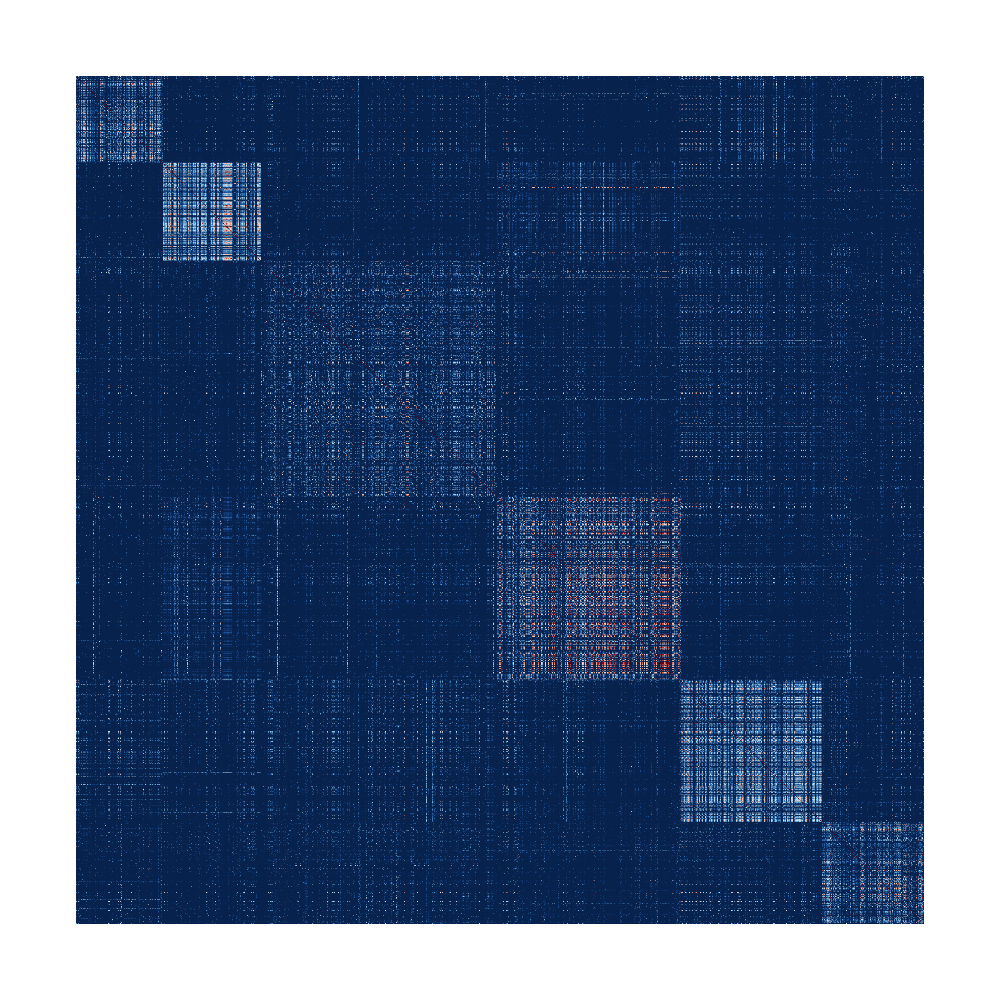}
    \caption{Kernel matrix of $\kappa$.}
    \label{fig:sfig2}
  \end{subfigure}
  \begin{subfigure}{.5\textwidth}
    \centering
    \includegraphics[width=.8\linewidth]{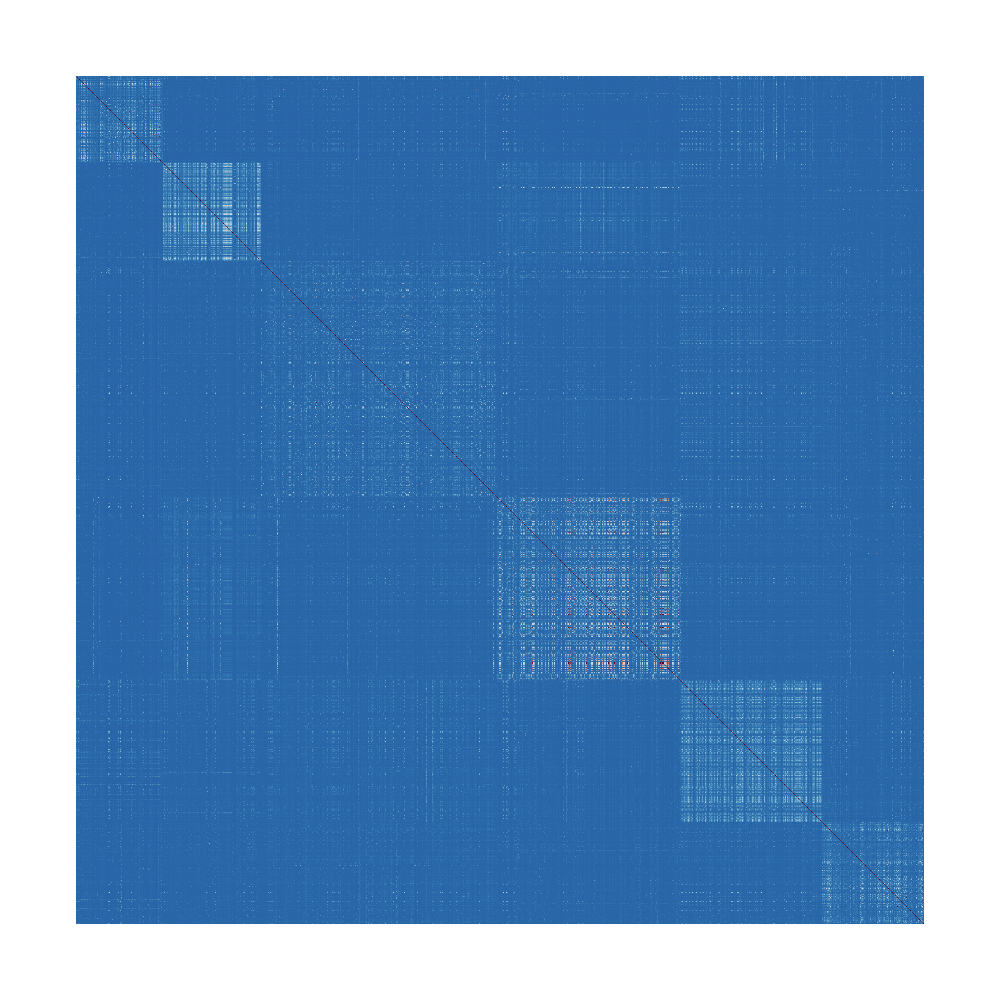}
    \caption{Exponential kernel matrix.}
    \label{fig:sfig5}
  \end{subfigure}
  \begin{subfigure}{.5\textwidth}
    \centering
    \includegraphics[width=.8\linewidth]{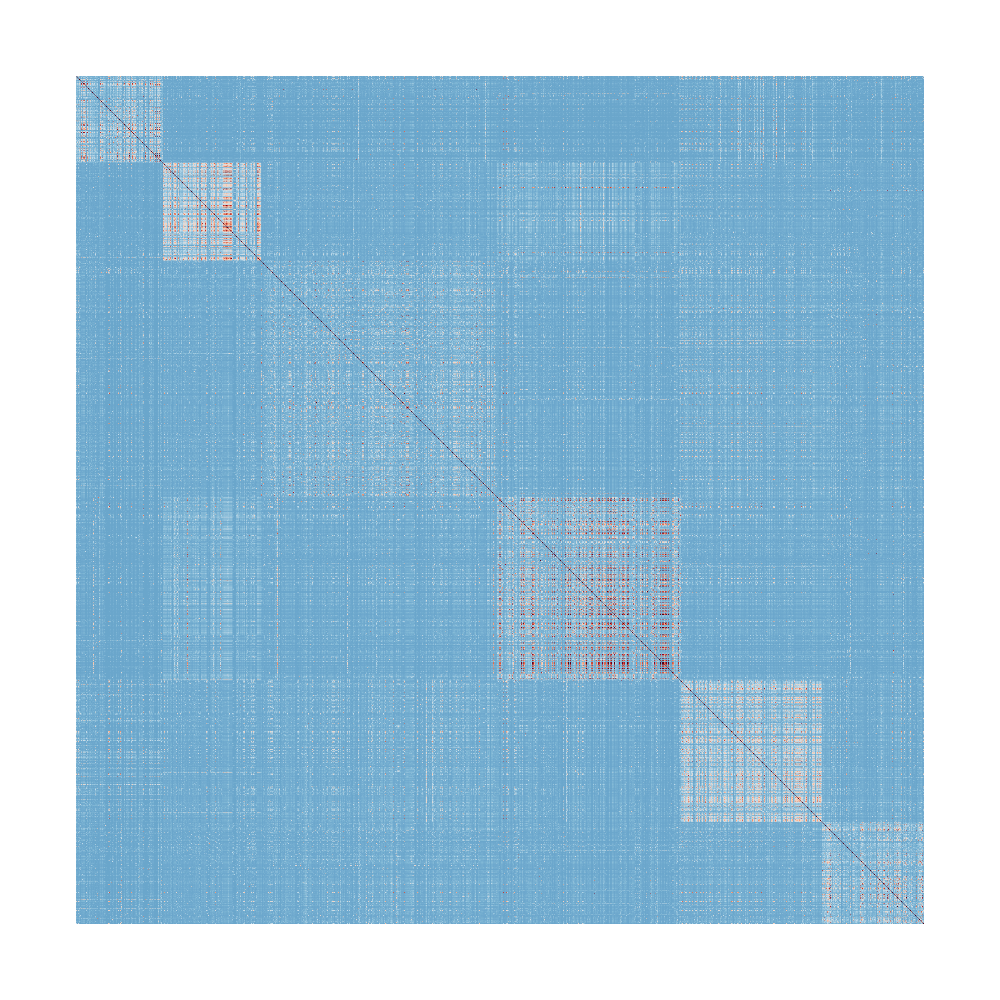}
    \caption{Radial basis function kernel matrix.}
    \label{fig:sfig6}
  \end{subfigure}
  \begin{subfigure}{.5\textwidth}
    \centering
    \includegraphics[width=.8\linewidth]{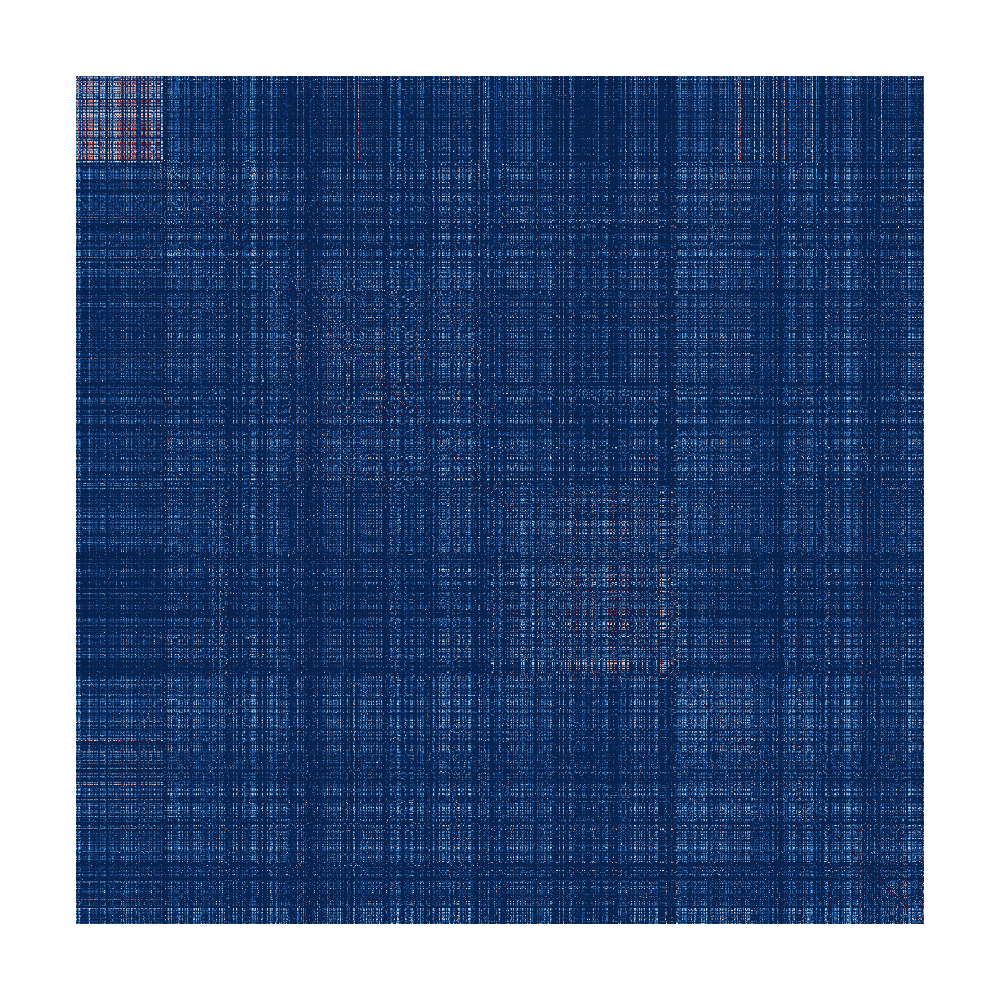}
    \caption{Spectrum kernel matrix.}
    \label{fig:sfig3}
  \end{subfigure}
  \begin{subfigure}{.5\textwidth}
    \centering
    \includegraphics[width=.8\linewidth]{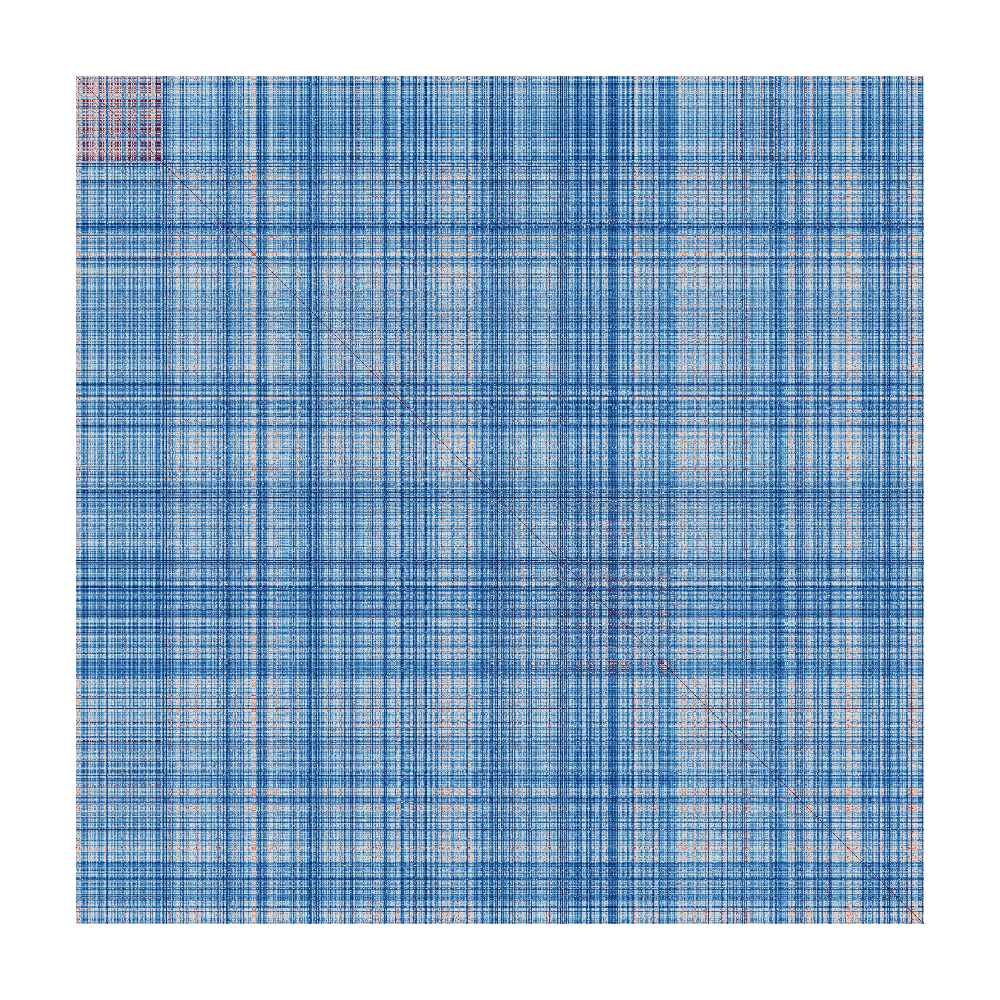}
    \caption{Boundrange kernel matrix.}
    \label{fig:sfig4}
  \end{subfigure}
  \caption{{\small Heatmaps of kernel matrices for patients in the UIHC EHR data. The color palette represents 0 and 1 using dark blue and dark red colors, respectively. The first four kernel matrices in (a)--(d) set $\Lambdab = \Ib$,
    $\sigma^2=1$, $\phi=1$, and $w_c \propto \kappa_1(\Tb_c, \Tb_c)$ for kernels $\kappa_2, \kappa, \kappa_1^e$, and $\kappa_2^e$, where $\Tb_c$ is the set of ICD codes defining the $c$th chronic condition. The last two kernel matrices are obtained using spectrum and boundrange kernels, which are popular string kernels in text mining, with substring length equalling 3. The block correlation structures are better captured in (b), (c), and (d) than (a), (e), and (f), indicating the superiority of the proposed set of kernel functions.
  }%
  }
  \label{fig:fig}
\end{figure}

We illustrate the application of six kernel functions in capturing the similarity of patients in the UIHC EHR data, where we expect patients with common primary sites to have similar sets of chronic conditions (Figure \ref{fig:fig}). Let $\Tb_i$ be the ICD codes for the patent $i$, $n$ be the sample size, $k(\cdot, \cdot)$ be a kernel function on $(2^{\Tcal})^C \times (2^{\Tcal})^C$. Then, the kernel matrix defined using $k$ is an $n \times n$ matrix with $(i,j)$-th entry $k(\Tb_i, \Tb_j)$. To highlight the similarity of patients, we have grouped the patients into six blocks depending on their primary cancer sites. The kernel matrix with $k=\kappa_2$ in \eqref{eq:4} fails to capture the similarities of patients with common primary sites  (Figure \ref{fig:sfig1}). On the other hand, the kernel matrix with $k=\kappa$ in \eqref{eq:7}, which is the normalized version of $\kappa_2$, captures the patient similarities, where the diagonal blocks denote patient with common primary sites and the off-diagonal blocks capture similarities of patients with cancers at two different sites (Figure \ref{fig:sfig2}). The same observation is true for the $\gamma$-exponential kernel in \eqref{eq:thm-kern} with $\gamma$ equals 1 and 2, respectively (Figures \ref{fig:sfig5} and \ref{fig:sfig6}). The diagonal and off-diagonal blocks are clearest for radial basis function kernel. The spectrum and boundrange kernel matrices fail to capture similarities of patients except those with cancer at brain and nervous system, where the diagnoses are very similar compared to cancers at other sites (Figures \ref{fig:sfig3} and \ref{fig:sfig4}).

\section{Modeling With Subsets of ICD Codes as Covariates}

\subsection{Setup}
\label{mdl-setup}

Consider the setup for regression and classification problems with ICD codes as covariates. Let $n$ be the sample size and $(y_i, \xb_i, \Tb_i)$ be the data for subject $i$ ($i=1, \ldots, n$), where $y_i$ is the response, $\xb_i^\T = (x_{i1}, \ldots, x_{ip})$ is the vector of $p$-dimensional covariates excluding ICD codes with $x_{i1}=1$, and $\Tb_i = \{\tb_{i1}, \ldots, \tb_{iC}\} \in (2^{\Tcal})^C$. For every $i$, $\tb_{ic}$ denotes the diagnoses of subject $i$ that belong to chronic condition $c$ and $\tb_{ic}$ can be an empty set $(c=1, \ldots, C)$. In regression and classification problems, $y_i \in \RR$ and $y_i \in \{0,1\}$, respectively, and $(x_{i2}, \ldots, x_{ip})^{\T} \in \RR^{p-1}$ for every $i$. We also predict the responses and estimate the covariate effects for a given collection of covariates including ICD codes $(\xb_j^*, \Tb_j^*)$ $(j=1, \ldots, m)$, where $\Tb_j^* \in (2^{\Tcal})^C$. 

We model the effects of $\xb_i$ and $\Tb_i$ independently. Let $\betab \in \RR^{p}$ and $f:(2^{\Tcal})^C \mapsto \RR$. Then, the covariate effects of $\xb_i$ and $\Tb_i$ are $\xb^{\T}_i \betab$ and $f(\Tb_i)$, respectively, and the overall covariate effect for subject $i$ as $\xb^{\T}_i \betab + f(\Tb_i)$. In practice, $\betab$ and $f$ are unknown. Choosing a Bayesian approach, we estimate the posterior distributions of $\betab$ and $f$ given $\{(y_i, \xb_i, \Tb_i)\}_{i=1}^n$. We use a parametric model for the estimating the effects of $\xb_i$ because our focus in on non-parametric regression and classification using $\Tb \in (2^{\Tcal})^C$.
. 

\subsection{Regression Using Subsets of ICD Codes as Covariates}
\label{mdl-reg}

Consider the nonparametric regression model. Let $\epsilon_i$ denote the idiosyncratic error for subject $i$, $\epsilon_i$s are independent and identically distributed (iid) with mean 0, and $\tau^2 >0 $ denotes the error variance. Then, the model assumes that               
\begin{align}
  \label{eq:m1}  
  y_i = \xb_i^\T \betab + f(\Tb_i) + \epsilon_i, \quad \epsilon_i \overset{\text{iid}}{\sim} N(0, \tau^2), \quad i =1, \ldots, n, 
\end{align}
where $N(a, b)$ denotes the Gaussian density with mean $a$ and variance $b$. We put a GP prior on $f(\cdot)$ with 0 mean function and covariance function $\kappa_f(\cdot, \cdot \mid \thetab)$ on $(2^{\Tcal})^C \times (2^{\Tcal})^C$,
where $\thetab$ are the parameters specifying $\kappa_f$. The model is completed by putting a prior distribution on $(\betab, \tau^2, \thetab)$ and its form depends on the choice of $\kappa_f$.

We develop posterior inference algorithms for $\kappa_f$ equalling $\kappa_{\gamma}^e$ defined in Theorem \ref{thm1}. We exclude the polynomial kernel because it yields a finite dimensional regression model, which is inflexible compared to the non-parametric regression model obtained using $\kappa_{\gamma}^e$. Setting $\wb = (1, \ldots, 1)$ and $\Lambdab_b$ to be an identity matrix $(b \in \Bcal)$ in \eqref{eq:thm-feat} and \eqref{eq:thm-kern}, we define
\begin{align}
  \label{eq:m2}
  \kappa_{\text{exp}}(\Tb, \Tb') = \kappa_1^e(\Tb, \Tb'), \quad
  \kappa_{\text{SE}}(\Tb, \Tb') = \kappa_2^e(\Tb, \Tb'), \quad \Tb, \Tb' \in (2^{\Tcal})^C \times (2^{\Tcal})^C
\end{align}
The parameter for $\kappa_{\text{exp}}$ and $ \kappa_{\text{SE}}$ is $\thetab = \{\sigma^2, \phi\}$, where $\sigma^2, \phi$ are positive scalars. Denote the kernel matrix as $\Kb_{\thetab}$ with $\{\Kb_{\thetab}\}_{ij} = \kappa_{\text{exp}}(\Tb_i, \Tb_j)$ or $\{\Kb_{\thetab}\}_{ij} = \kappa_{\text{SE}}(\Tb_i, \Tb_j)$, where $i, j \in \{1, \ldots, n\}$. Let $\yb = (y_1, \ldots, y_n)^\top$, $\Xb$ be the $n \times p$ matrix with $\xb_i^\T$ as its $i$th row, $\fb = \{f(\Tb_1), \ldots, f(\Tb_n)\}^\T$, and $\epsilonb = (\epsilon_1, \ldots, \epsilon_n)^\T$. Then, the hierarchical Bayesian model for $\yb$ based on \eqref{eq:m1} and parameter-expansion is defined as 
\begin{align}
  \label{eq:8}
  &\yb = \Xb \betab + \fb + \epsilonb, \quad \epsilonb \mid \sigma^2, \alpha \sim N(\zero, \sigma^2 \alpha \Ib), \quad  (\betab, \sigma^2) \propto \sigma^{-2}, \quad \fb \mid \thetab \sim N(\zero, \Kb_{\thetab}),
\end{align}
where $\betab$ and $\fb$ are assumed to be independent apriori, $\alpha = \tau^2/\sigma^2$ is the inverse of signal-to-noise ratio, and $\phi \in (a_{\phi}, b_{\phi})$, where $0 < a_{\phi} < b_{\phi}$.  We assign a prior on $(\phi, \alpha)$ through $(u_1, u_2) \in \RR^2$ as 
\begin{align}
  \label{eq:81}
  \phi = a_{\phi} + (b_{\phi} - a_{\phi}) / (1 + e^{-u_1}), \quad \alpha = e^{u_2}, \quad (u_1, u_2)^\top \sim N\{\zero, \diag(b_1, b_2)\}, 
\end{align}
where $b_1, b_2 > 0$. The parameters $\sigma^2,\phi$ are non-identified in $\kappa_{\text{exp}}, \kappa_{\text{SE}}$, but this  does not affect the inference on $f$ or prediction of the response \citep{RasWil06}.

The MCMC algorithm for posterior inference and predictions in \eqref{eq:8} is a Gibbs-type sampling algorithm, with a step that draws $(\phi, \alpha)$ using Elliptical Slice Sampling (ESS) \citep{Nisetal14}. Let $\fb^* = \{f(\Tb_1^*), \ldots, f(\Tb^*_m)\}^\T$ be the value of $f$ and $\yb^* = (y_1^*, \ldots, y_m^*)^\T$ be the response at $(\xb_j^*, \Tb_j^*)$ $(j=1, \ldots, m)$. For notational convenience, define $\Kb(\phi)$, $\Kb_*$, and $\Kb_{**}$ to be the matrices satisfying $\{\Kb(\phi)\}_{ii'} = \kappa_f(\Tb_i, \Tb_{i'}) / \sigma^2$, $\{\Kb_*\}_{ij} = \kappa_f(\Tb_i, \Tb^*_j)/ \sigma^2$, and $\{\Kb_{**}\}_{jj'} = \kappa_f(\Tb_j^*, \Tb^*_{j'})/ \sigma^2$, where $i, i' \in \{1, \ldots, n\}$ and $j, j' \in \{1, \ldots, m\}$, and $\Cb_{yy} = \Kb(\phi) + \alpha \Ib$. If $\Lb$ is a lower triangular matrix such that $\Cb_{yy} = \Lb \Lb^\top$, then define
\begin{align}
  \label{eq:82}
  \tilde \yb = \Lb^{-1} \yb, \quad \tilde \Xb = \Lb^{-1} \Xb, \quad \tilde \etab = \Lb^{-1} \etab, \quad \hat \betab = (\tilde \Xb^\T \tilde \Xb)^{-1} \tilde \Xb^\top \tilde \yb, \quad \hat{\tilde \yb} = \tilde \Xb \hat \betab. 
\end{align}
The MCMC algorithm for posterior inference on $\betab, \fb^*, \Kb_{**}$ and for predicting $\yb^*$ cycles through the following four steps until convergence to its stationary distribution.
\begin{enumerate}
\item Draw $(\sigma^2, \betab)$ given $\yb, \phi, \alpha$ as 
  \begin{align}
    \label{eq:11}
    \sigma^2 \mid \yb, \phi, \alpha \sim  \frac{\|\tilde \yb - \hat{\tilde \yb} \|^2}{\chi^2_{n-p}}, \quad
    \betab \mid \sigma^2, \yb, \phi, \alpha \sim N\{\hat \betab, \sigma^2 (\tilde \Xb^\T \tilde \Xb)^{-1}\} ,
  \end{align}
  where $\chi^2_{n-p}$ is the $\chi^2$ random variable with $n-p$ degrees of freedom.  
\item Draw $(u_1, u_2)$  given $\yb, \sigma^2, \betab$ are drawn using ESS with the prior $(u_1, u_2) \sim N\{\zero, \diag(b_1, b_2)\}$ and the likelihood
  \begin{align}
    \label{eq:13}
    \ell(u_1, u_2) = \frac{1}{2\pi \sigma^2 [\text{det}\{ \Kb(\phi) +  \alpha \Ib\}]^{1/2}}e^{- \tfrac{1}{2 \sigma^{2}}(\yb - \Xb \beta)^\T \{\Kb(\phi) + \alpha \Ib\}^{-1} (\yb - \Xb \beta)}.
  \end{align}
  Define $\phi$ and $\alpha$ given $(u_1, u_2)$  as in \eqref{eq:81} and set $\tau^2 = \sigma^2 \alpha$.
\item Draw $\fb^*$ given  $\yb, \sigma^2, \betab, \phi, \alpha$ from $N(\mb_*, \Vb_*)$, where
  \begin{align}
    \label{eq:141}
    \mb_* = \Kb_*^\top \Cb_{yy}^{-1} (\yb - \Xb \betab), \quad 
    \Vb_* &= \sigma^2 \left( \Kb_{**} - \Kb_{*}^\top \Cb_{yy}^{-1} \Kb_{*} \right). 
  \end{align}
\item Draw $\yb^*$ given $\Xb^*, \betab, \fb^*, \tau^2$ from $N(\Xb^* \betab + \fb^*, \tau^2 \Ib)$.
\end{enumerate}

We collect MCMC draws of $\betab, \fb^*, \yb^*, \sigma^2, \phi, \tau^2$ post convergence. The derivation of this algorithm is in the supplementary materials along with the others. The algorithm in steps (1)--(4) a variant of Gibbs sampling algorithm for posterior inference in univariate spatial linear models in that we replace the Metropolis-Hastings step by an ESS step in (2) and $\Tb$ replaces a spatial location \citep{Banetal14}. Unlike the Metropolis-Hastings step, the ESS step is free of any proposal tuning, which is preferred in automated applications.  

\subsection{Classification Using Subsets of ICD Codes as Covariates}
\label{mdl-class}

The classification model is based on logistic regression. It assumes that 
\begin{align}
  \label{eq:log1}
  \log \, \frac{\Pr (y_i = 1)}{\Pr (y_i = 0)} = \xb_i^\T \betab + f(\Tb_i), \quad y_i \in \{0, 1\}, \quad i =1, \ldots, n,   
\end{align}
where $y_i$ is the $i$th response. The choice of priors on $\betab, f(\cdot)$ and the kernels remain the same as in \eqref{eq:8}, but the prior on $\thetab = (\phi, \sigma^2)^\top$ is assigned through $(u_1, u_2) \in \RR^2$ as 
\begin{align}
  \label{eq:81log}
  \phi = a_{\phi} + (b_{\phi} - a_{\phi}) / (1 + e^{-u_1}), \quad \sigma^2 = e^{u_2}, \quad (u_1, u_2)^\top \sim N\{\zero, \diag(b_1, b_2)\}. 
\end{align}
where $b_1, b_2 > 0$. We use P\'olya-Gamma data augmentation (PG-DA) for posterior inference on $\betab, \fb^*, \Kb_{**}$ and prediction of $\yb^*$ \citep{Poletal13}. This setup ensures that the MCMC algorithms for inference and predictions in classification and regression models are very similar. Specifically, we introduce P\'olya-Gamma random variables $\omega_1, \ldots, \omega_n$ specific to every observation such that $\omega_i$ are marginally distributed as PG(1, 0), where PG is the P\'olya-Gamma distribution with parameters $b=1$ and $c=0$, respectively. Define $\bar{\yb} = (y_1 - 1/2, \ldots, y_n - 1/2)^{\T}$, $n \times (p + n)$ matrix $\Ab = [\Xb \; \Ib]$, $(p+n) \times 1$ vector $\bb = (\betab, \fb)^\top$, $n \times 1$ vector $\omegab = (\omega_1, \ldots, \omega_n)^\T$, $n \times n$ diagonal matrix $\Omegab = \diag(\omegab)$, pseudo responses $\zb = \Omegab^{-1} \bar \yb$. Then, the conditional likelihood of $\bb$ given $\yb, \omegab$ and the associated model are defined as 
\begin{align}
  \label{eq:log2a}
  \ell (\etab \mid \yb, \omegab) \propto e^{ - \frac{1}{2} (\Ab \bb - \zb)^\top \Omega (\Ab \bb - \zb)}, \quad \zb = \Ab \bb + \epsilonb = \Xb \betab + \fb + \epsilonb, \quad \epsilonb \sim N(0, \Omegab^{-1}),
\end{align}
respectively. Furthermore, Theorem 1 in \citet{Poletal13} implies that $\omega_i$ given $\bb$ and the $i$th row of $A$, denoted as $\ab_i^\T$, follows PG(1, $|\ab_i^\T \bb|$).

The MCMC algorithm for posterior inference and predictions in \eqref{eq:log1} follows from arguments similar to those used in Section \ref{mdl-reg}. Marginalizing over $\fb$ in \eqref{eq:log2a} implies that $\zb $ given $\yb$, $\betab$, $\omegab$, $\phi, \sigma^2$ is distributed as $N(\Xb \betab, \Cb_{zz})$, where $\Cb_{zz} = \Kb_{\thetab} + \Omegab^{-1}$ and it plays the same role as $\Cb_{yy}$ in \eqref{eq:8}. Following \eqref{eq:82}, if $\Lb$ is a lower triangular matrix such that $\Cb_{zz} = \Lb \Lb^\top$, then define
\begin{align}
  \label{eq:821}
  \tilde \zb = \Lb^{-1} \zb, \quad \tilde \Xb = \Lb^{-1} \Xb, \quad \hat \betab = (\tilde \Xb^\T \tilde \Xb)^{-1} \tilde \Xb^\top \tilde \zb. 
\end{align}
The MCMC algorithm for posterior inference on $\betab, \fb^*, \Kb_{**}$ and for predicting $\yb^*$ cycles through the following six steps until convergence to its stationary distribution.
\begin{enumerate}
\item Draw $\betab$ given $\yb, \omegab, \phi, \sigma^2$ as $N\{\hat \betab, (\tilde \Xb^\T \tilde \Xb)^{-1}\}$.
\item Draw $(u_1, u_2)$  given $\yb, \omegab, \betab$ using ESS following \eqref{eq:13}, where the likelihood and prior for $(u_1, u_2)$ are defined in \eqref{eq:log2a} and \eqref{eq:81log}. Define $\phi$ and $\sigma^2$ given $(u_1, u_2)$ using \eqref{eq:81log}. 
\item Draw $\fb$ given $\yb, \omegab, \betab, \phi, \sigma^2$ as $N(\mb, \Vb)$, where
  \begin{align}
    \label{eq:log-mcmc-2}
    \Vb = \Kb_{\thetab} - \Kb_{\thetab} \Cb_{zz}^{-1} \Kb_{\thetab}, \quad
    \mb = \Vb \Omegab (\zb - \Xb \betab).
  \end{align}
\item Draw $\fb^*$ given  $\yb, \betab, \sigma^2, \phi$ from $N(\mb_*, \Vb_*)$, where
  \begin{align}
    \label{eq:14}
    \mb_* = \Kb_{\thetab *}^\top \Cb_{zz}^{-1} (\zb - \Xb \betab), \quad 
    \Vb_* &= \Kb_{\thetab **} - \Kb_{\thetab *}^\top \Cb_{zz}^{-1} \Kb_{\thetab *} ,
  \end{align}
  where $\{\Kb_{\thetab *}\}_{ij} = \kappa_f(\Tb_i, \Tb^*_j)$, and $\{\Kb_{\thetab **}\}_{jj'} = \kappa_f(\Tb_j^*, \Tb^*_{j'})$, where $i \in \{1, \ldots, n\}$ and $j, j' \in \{1, \ldots, m\}$.
\item Draw $y^*_i$ given $\Xb^*, \betab, \fb^*$ from Bernoulli($p_i^*$) for $i = 1, \ldots, m$, where $p^*_i = e^{v_i} / (1 + e^{v_i} )$ and $v_i = \xb_i^{*\top} \betab + (\fb^*)_i$. 
\item Draw $\omega_i$ given $\Xb, \betab, \fb$ from PG$\{1, |\xb_i^\top \betab + (\fb)_i|\}$ for $i = 1, \ldots, n$.
\end{enumerate}
We collect MCMC draws of $\betab, \fb^*, \yb^*, \sigma^2, \phi$ post convergence.

Our MCMC algorithm is similar to other algorithms based on the PG-DA strategy. First, \citet[Section 5, Table 3]{Poletal13} present an application of the PG-DA strategy for nonparametric regression with a negative binomial response and a GP prior, but a linear predictor similar to $\xb_i^\top \betab$ in \eqref{eq:m1} is absent in their model. Second, \citet{WanRoy18b} develop a sampling algorithm based on the PG-DA strategy for posterior inference in a Bayesian logistic linear mixed model with independent and Gaussian random effects and prove its geometric ergodicity. Our MCMC algorithm is similar to theirs in that $f(\cdot)$ plays the role of random effects and the GP prior on $f(\cdot)$ induces dependence among patients that have similar diagnoses. Overall, the MCMC algorithm above broadens the range of application of the PG-DA strategy. 

\subsection{Theoretical Properties}
\label{sub-sec:theory}

The convergence rates of the posterior distributions of $f(\cdot)$ in Sections \ref{mdl-reg} and \ref{mdl-class} follow from the general theoretical setup for GP regression and classification \citet[Section 11.4.4]{GhoVan17}. Our focus is on the stationary GPs with the Mat\'ern and SE covariance kernels because $\kappa_\gamma^e$ is based on them. For stationary GP priors, the posterior convergence rates are described using  tail-decay properties of the \emph{spectral measure} associated with the covariance kernel and \emph{regularity} of the true $f(\cdot)$. In this section, we adapt existing results for our feature maps $\overline \Psib_{\wb, \Rb}$ assuming that $\wb$ and $\Rb$ are known.

Consider the Mat\'ern GP prior and its spectral measure. Following \citet[Page 84]{RasWil06}, if $d = C|t^*||\Bcal|$, then the Mat\'ern covariance kernel with smoothness and range parameters $\nu$ and $\phi$, respectively, is defined through its spectral measure $\lambda$ on $\RR^d$ as
\begin{align}
  \label{eq:sp-mdl}
  \kappa_{\text{Mat}}(\hb \mid \nu, \phi) &= \int_{\RR^{d}} e^{- \iota 2 \pi  \hb^\top \sbb} d\lambda(\sbb), \quad \hb \in \RR^{d}, \quad \nu > 0, \quad \phi > 0,\quad \iota = \sqrt{-1}, \nonumber\\
  d\lambda(\sbb) &= \frac{2^d \pi^{d/2}\Gamma(\nu + d/2)}{\Gamma(\nu)} (2\nu \phi^{2})^{\nu} \left( 2 \nu \phi^{2} + 4 \pi^2 \| \sbb \|^2 \right)^{-\nu - d/2} d \sbb, \quad \sbb \in \RR^d, \nonumber\\
                   \kappa_{\text{Mat}}(\hb \mid \nu, \phi)  &= \frac{2^{1-\nu}}{\Gamma(\nu)} \left( \sqrt{2\nu} \phi\| \hb \| \right)^{\nu} \Kcal_{\nu} \left(\sqrt{2 \nu} \phi\| \hb \| \right),              
\end{align}
where $\Kcal_{\nu} $ is the modified Bessel function of second kind. If $\nu \in \NN$, then sample paths of the Mat\'ern GP have partial derivatives up to order $\nu$ and all of them are square integrable. 

The GP prior in Theorem \ref{thm1} is indexed by the feature vectors $\overline \Psib_{\wb,\Rb}(\Tb) \in [0, 1]^d$ for $\Tb \in (2^{\Tcal})^C$. This is equivalent to embedding the string $\Tb$ in $[0, 1]^d$ using the feature map. The GP covariance kernels are defined using this embedding. The equivalent of Mat\'ern covariance kernel on $(2^{\Tcal})^C \times (2^{\Tcal})^C$ with variance $\sigma^2$ is defined through $\overline \Psib_{\Tb} ,  \overline \Psib_{\Tb'}$ as $\sigma^2 \kappa_{\text{Mat}}(\overline \Psib_{\Tb} - \overline \Psib_{\Tb'} \mid \nu, \phi)$, where $\overline \Psib_{\Tb}$ denotes $\overline \Psib_{\wb,\Rb}(\Tb)$. If $\nu=1/2$, then this kernel equals $\kappa_1^e$. Taking point-wise limit of $\sigma^2 \kappa_{\text{Mat}}(\overline \Psib_{\Tb} - \overline \Psib_{\Tb'} \mid \nu, \phi)$ as $\nu \rightarrow \infty$, we obtain the equivalent of SE covariance kernel on $(2^{\Tcal})^C \times (2^{\Tcal})^C$ as $\sigma^2 \exp({ - \phi^2 \| \overline \Psib_{\Tb} - \overline \Psib_{\Tb'}\|^2 / 2 })$. 


We now define the regularity of functions using three function spaces. Let $\lfloor \nu \rfloor$ denote the largest integer less than or equal to $\nu$ and $\check g(\sbb)$ denote the Fourier transform of  $g:\RR^d \rightarrow \RR$. First, the H\"older space $\mathfrak{C}^{\nu}[0, 1]^d$ is the space of all functions on $[0, 1]^d$ whose partial derivatives of order $(k_1, \ldots, k_d)$ exist for all nonnegative integers $k_1, \ldots, k_d$ such that $k_1 + \ldots + k_d \leq \lfloor  \nu \rfloor $ and whose $\lfloor  \nu \rfloor$-th derivative is Lipschitz of order $\nu - \lfloor  \nu \rfloor$.
Second, the Sobolev space $\mathfrak{H}^{\nu}[0, 1]^d$ is space of all functions on $[0, 1]^d$ that are the restrictions of a function $g$ for which $\int (1 + \| \sbb\|^2)^{\nu} |\check g(\sbb)|^2 d \sbb $ is finite. Finally, the space of infinitely smooth functions on $[0,1]^d$ are restrictions of a function $g$ on  $[0,1]^d$ for which $\int \exp (\gamma \| \sbb \|^r) |\check g(\sbb)|^2 d \sbb $ is finite and is denoted as $\mathfrak{A}^{r,\gamma}[0,1]^d$ with  $r \geq 1$ and $\gamma > 0$. A function is $\nu$-regular and $\infty$-regular on $[0, 1]^d $ if it belongs to  $\mathfrak{C}^{\nu}[0, 1]^d \cap \mathfrak{H}^{\nu}[0, 1]^d$ and $\mathfrak{A}^{r,\gamma}[0, 1]^d$ for some $r \geq 1$ and $\gamma > 0$, respectively.

We introduce some notations for stating the theoretical results.  Reformulating the models in Sections \ref{mdl-reg} and \ref{mdl-class} in terms of feature maps $\Psi \in [0, 1]^d$, the population versions of \eqref{eq:m1} and \eqref{eq:log1} after setting $\betab = \zero$ are 
\begin{align}
  \label{eq:th-mdl}
  y = \mu_0(\Psi) + \epsilon, \; \epsilon \sim N(0, \tau^2_0), \; \tau^2_0 > 0, \quad \log \, \frac{\Pr (y = 1)}{\Pr (y = 0)} = \mu_0(\Psi),  \quad 
\end{align}
respectively, where $\tau_0^2$ is the error variance, any $\Tb \in (2^{\Tcal})^C$ is embedded in $[0, 1]^d$ as $\Psi = \overline \Psib_{\wb, \Rb}(\Tb)$, and the effect of $\Tb$ on $y$ is modeled through the feature map $ \overline \Psib_{\wb, \Rb}(\Tb)$, which is fixed for a given $\wb, \Rb$. The model is \eqref{eq:th-mdl} is well-defined because $\Psi$ is defined in Theorem \ref{thm1} and is fixed if $\wb, \Rb$ are known. The true relation between the feature map and response is denoted as $\mu_0(\cdot)$, which maps  $[0, 1]^d$ to $\RR$. The training data $y_1, \ldots, y_n$ are independently generated following \eqref{eq:th-mdl} in the two models for fixed $\Tb_1, \ldots, \Tb_n$ that have been determined apriori. Every  $\Tb_i$ is embedded in $[0, 1]^d$ as $\Psi_i = \overline \Psib_{\wb, \Rb}(\Tb_i)$ and \eqref{eq:th-mdl} is satisfied for all the $n$ embeddings. The true distribution of $y_i$ is $N \{\mu_0(\Psi_i), \tau_0^2\}$ and Bernoulli[$1 / \{1 + e^{-\mu_0(\Psi_i)}\}$] in \eqref{eq:m1} and \eqref{eq:log1}, respectively. The $n$-fold product of these distributions is denoted as $\PP_0^n$ in both models and $\yb \sim \PP_0^n$. The GP prior on $f(\cdot)$ in \eqref{eq:m1} and \eqref{eq:log1} is induced using the Mat\'ern and SE covariance kernels if $\mu_0(\cdot)$ in \eqref{eq:th-mdl} has a finite or infinite regularity, respectively. The posterior distributions of $f$ given $\yb$ are denoted as $\Pi_{\text{Mat}, n}^{\text{Reg}}$, $\Pi_{\text{SE}, n}^{\text{Reg}}$ and $\Pi_{\text{Mat}, n}^{\text{Clas}}$, $\Pi_{\text{SE}, n}^{\text{Clas}}$ in  \eqref{eq:m1} and \eqref{eq:log1} depending on the GP prior.

The convergence rates are described in terms of posterior probabilities assigned to small neighborhoods of $\mu_0$ as $n \rightarrow \infty$. Using  \eqref{eq:th-mdl} and the embedding of $(2^{\Tcal})^C$ in $[0, 1]^d$ based on the feature map $\overline \Psib_{\wb, \Rb}$, the distance between $f(\cdot)$ and $\mu_0(\cdot)$ is defined as the empirical norm
\begin{align}
  \label{eq:th-met}
  \| f - \mu_0 \|_n =  \left( \frac{1}{n} \sum_{i=1}^n \left[ f \left\{ \overline \Psib_{\wb, \Rb}(\Tb_i) \right\} - \mu_0\left\{ \overline \Psib_{\wb, \Rb}(\Tb_i) \right\}  \right]^2 \right)^{1/2} , 
\end{align}
where $\| \cdot \|_n$ conditions on the covariates $\Tb_1, \ldots, \Tb_n$ and its dependence on $\wb$ and $\Rb$ is suppressed. Given $\wb, \Rb$, the true and estimated effects of $\Tb_i$ on $y_i$ are  $\mu_0\left\{ \overline \Psib_{\wb, \Rb}(\Tb_i) \right\} $ and $f \left\{ \overline \Psib_{\wb, \Rb}(\Tb_i) \right\} $, respectively. The next theorem quantifies the rate of convergence of the posterior distribution in terms of $\| \cdot \|_n$ norm. 
\begin{theorem}\label{thm2}
  Assume that the regression and classification models defined in \eqref{eq:th-mdl} are true, $\sigma$ and $\phi$ are known, and the Mat\'ern covariance kernel has smoothness $\nu_1$. 
  \begin{enumerate}
  \item If $\mu_0$ is $\nu_2$-regular in the sense that  $\mu_0 \in \mathfrak{C}^{\nu_2}[0, 1]^d \cap \mathfrak{H}^{\nu_2}[0, 1]^d$ with $\nu_2 \leq \nu_1$, then as $n \rightarrow \infty$ there are positive constants $C_1, C_2$ such that 
    \begin{align}
      \label{eq:mat-rate}
      \Pi_{\text{Mat}, n}^{\text{Reg}} \left\{ f : \| f - \mu_0 \|_{n}  > C_1 \epsilon_n \mid \yb \right\} \rightarrow 0, \quad
      \Pi_{\text{Mat}, n}^{\text{Clas}} \left\{ f : \| f - \mu_0 \|_{n} > C_2 \epsilon_n \mid \yb \right\} \rightarrow 0
    \end{align}    
    in $\PP_0^n$-probability for every $\wb, \Rb$ with $\epsilon_n = n^{- \nu_2 / (2 \nu_1 + d)}$. 
  \item If $\mu_0$ is infinitely smooth functions in the sense that  $\mu_0 \in \mathfrak{A}^{r,\gamma}[0, 1]^d$ for some $r \geq 1$ and $\gamma > 0$, then as $n \rightarrow \infty$
    there are positive constants $C_1, C_2$ such that 
    \begin{align}
      \label{eq:SE-rate}
      \Pi_{\text{SE}, n}^{\text{Reg}}  \left\{ f : \| f - \mu_0 \|_{n}  > C_1 \epsilon_n \mid \yb \right\} \rightarrow 0, \quad 
      \Pi_{\text{SE}, n}^{\text{Clas}} \left\{ f : \| f - \mu_0 \|_{n} > C_2 \epsilon_n \mid \yb \right\} \rightarrow 0
    \end{align}
    in $\PP_0^n$-probability for every $\wb, \Rb$ with $\epsilon_n = n^{-1/2}(\log n)^{\max (1/r, 1/2 + d /4 )} $. 
  \end{enumerate}     
\end{theorem}

The proof of Theorem \ref{thm2} is based on the standard setup \citep[Theorems 11.23 and 11.22]{GhoVan17}. The theorem is also applicable to the case when $\betab $ is non-zero if we absorb $\xb$ as a part of the $\overline \Psib_{\Tb}$ vector. Specifically, we scale the covariates $\xb, \xb'$ so that they lie in $[0, 1]^{p}$, redefine $d $ to be $p + C|t^*||\Bcal|$, and set the Mat\'ern and SE kernels to be $\sigma^2 \kappa_{\text{Mat}}(\overline \Psib_{\Tb} - \overline \Psib_{\Tb'} \mid \nu, \phi) \cdot  \kappa_{\text{Mat}}(\xb - \xb' \mid \nu, \phi)$ and $\sigma^2 \exp({ - \phi^2 \| \overline \Psib_{\Tb} - \overline \Psib_{\Tb'}\|^2 / 2 }) \cdot \exp({ - \phi^2 \| \xb - \xb' \|^2 / 2 })$, respectively, for any $(\xb, \Tb), (\xb', \Tb')$ in $[0,1]^{p} \times (2^{\Tcal})^C$. 

\section{Experiments}

\subsection{Setup}
\label{sec:setup}

We evaluate the performance of GP-based classification in \eqref{eq:log1} using three simulation studies and an analysis of the UIHC EHR data. In this section, covariance kernel of the GP is SE and the focus is on classification because our main goal is to address questions relevant to the UIHC EHR data; supplementary materials contain the results for GP-based regression using \eqref{eq:m1}. In the simulated and real data analyses, the cutoff for predicting the response is estimated using Receiver Operating Characteristics (ROC) curve. We compare the predictive performance of the all the methods using accuracy, area under the ROC curve (AUC), sensitivity, and specificity on the test data.

We compare GP-based classification with parametric and nonparametric models. The parametric competitors are logistic regression and its regularized versions that penalize the regression coefficients using the lasso and ridge penalties. The covariates in these models are dummy variables that indicate the presence of ICD codes in a patient's diagnosis. Because the SE kernel in \eqref{eq:7} is a restriction of the boundrange kernel, we use support vector machine (SVM) and kernel ridge regression (KRR) with the spectrum and boundrange kernels as the nonparametric competitors. We also use the kernel function in \eqref{eq:7} for defining feature vectors. Set $\phi=1, \sigma^2=1$ in the kernel matrix $\Kb_{\thetab}$. If $\Ub \Db \Ub^\top$ is the singular value decomposition of $\Kb$, then the first $r$ columns of $\Ub \Db^{1/2}$ are the feature vectors that are used as covariates in random forest, logistic regression, and its penalized extensions, where $r$ is selected using the scree plot.

The experiments are performed in R. We used glmnet \citep{Frietal09} for regularized logistic regression, ranger \citep{Wright2017} for random forest, and  kernlab \citep{Karetal04} for SVM and KRR. The performance metrics are evaluated using the pROC package \citep{Xavetal11}. The tuning parameters in all the methods are selected using the recommended settings. We use the spectrum and boundrange kernels in kernlab with 3 and 4 as the  length of the matching substring, respectively. This means that the matches in the two kernels include the first three and four contiguous characters in a pair of ICD codes, which are also the most informative \citep{Caletal17}; therefore, we expect the results for kernlab and our method to be very similar. We use ESS by setting $a_{\phi}$, $b_{\phi}$, $b_1$, and $b_2$ in \eqref{eq:81log} to be 0,5, 1, and 2, respectively. The MCMC algorithms run for 10,000 iterations. The draws from the first 5,000 iterations are discarded as burn-ins, and every fifth draw in the remaining chain is chosen for posterior inference and prediction. 

\subsection{Simulated data analysis}
\label{sec:simul-data-analys}

We present three simulation studies. The first two are based on simple parametric models, whereas the third is based on a nonparametric model based on the UIHC EHR data.  We expect all methods to have comparable performances in the first two simulations and differ significantly in the third one. The three simulations are replicated ten times.

\emph{\textul{First two simulations.}} Both simulations have four chronic conditions. The first chronic condition is denoted using the codes in $\{A, B, AA, BB, BA, AB\}$, where every code denotes a ``diagnosis''. Replacing $(A, B)$  with $(C, D)$, $(E, F)$, and $(G, H)$ in this set gives the set of codes for the second, third, and fourth chronic conditions, respectively. The chronic conditions are further structured into two groups: the first group includes the first and third chronic conditions and the second group includes the remaining two. We assign nine codes to every patient, where six out of the nine codes are from the 12 codes defining first or second group of chronic conditions and the remaining three come from the other chronic condition group. The probability of $y=1$ is higher than that of $y=0$ if codes from the second chronic condition group are in majority and vice-versa.

The first simulation study is based on logistic regression. Let $z_{ij}^{(c)}$ be a dummy variable indicating the presence of the $j$th code in the $c$th chronic condition in the $i$th patient ($c=1, \ldots, 4$; $j=1, \ldots, 6$; $i = 1, \ldots, n+m$), where codes follow the dictionary order, $n$ is the training data size, and $m$ testing data size. The first and second group dummy variables for the $i$th patient are $\{z_{ij}^{(c)}: j=1, \ldots, 6; c = 1, 3\}$ and $\{z_{ij}^{(c)}: j=1, \ldots, 6; c = 2, 4\}$. The response $y_i$ is simulated independently from Bernoulli ($p_i$), where
\begin{align}
  \label{eq:sim11}
  \log \frac{p_i}{1 - p_i} = 0.1 +  0.2 x_i + \sum_{c=1}^4   \sum_{j=1}^6 (-1)^c  z_{ij}^{(c)}, \quad x_i \overset{\text{ind.}}{\sim} N(0,1)  
\end{align}
for $i = 1, \ldots, n+m$. The second simulation is a slight modification of the first. Dropping the $\xb_i^\top \betab$ term in \eqref{eq:sim11}, we modify it to only include the interaction of $\{AB, BA\}$, $\{CD, DC\}$, $\{EF, FE\}$, and $\{GH, HG\}$, respectively, as
\begin{align}
  \label{eq:sim21}
  \log \frac{p_i}{1 - p_i} = 2 \left( z_{i \text{EF}}^{(2)} \cdot z_{i \text{FE}}^{(2)}  + z_{i \text{GH}}^{(4)} \cdot z_{i \text{HG}}^{(4)} \right) - \left( z_{i \text{AB}}^{(1)} \cdot z_{i \text{BA}}^{(1)} + z_{i \text{CD}}^{(3)} \cdot z_{i \text{DC}}^{(3)} \right),  
\end{align}
and $y_i = 1(p_i > 0.5)$  for $i = 1, \ldots, n+m$.

In first two simulations, GP with SE covariance kernel is among the best performers (Table \ref{tab:sim1}). The first simulation has a linear decision boundary in terms of the 24 codes. On the other hand, the second simulation has a quadratic decision boundary depending on the interactions between $\{AB, BA\}$, $\{CD, DC\}$, $\{EF, FE\}$, and $\{GH, HG\}$. Logistic regression and its regularized versions perform the best in the first simulation, but their performance deteriorates slightly in the second simulation due to the presence of interactions between dummy variables. Random forest using the codes as covariates has the worst performance in both simulations. After including SE kernel-based covariates, random forest achieves the same level of performance as other method in both simulations. The remaining methods are based on non-linear kernels, so they are more flexible in adapting to both linear and non-linear decision boundaries. Compared to the spectrum and boundrange kernels, the SE kernel in \eqref{eq:7} is better tuned for modeling the extra structure provided by the chronic conditions. This implies that kernel-based methods perform well in both simulations and that the GP with SE kernel performs better than SVM and KRR with spectrum and boundrange kernels.

\emph{\textul{Third simulation.}} For this simulation, we used ICD codes of patients in the UIHC EHR data with brain and other nervous system or breast as the cancer primary sites. For the $i$th (pseudo) patient in this subset, let 
$z_i = \| \Psib (\Tb_i) \|^2$ and $\delta_i=1$ if the patient with $\Tb_i$ code has brain and other nervous system cancer and $\delta_i=-1$ otherwise, where $\Psib (\Tb_i)$ is defined in \eqref{eq:3}. The response $y_i$ is simulated independently from Bernoulli ($p_i$), where
\begin{align}
  \label{eq:sim31}
  \log \frac{p_i}{1 - p_i} = 0.1 +  0.2 x_i + 3\tan(z_i) + 3\delta_i, \quad x_i \overset{\text{ind.}}{\sim} N(0,1) 
\end{align}
for $i = 1, \ldots, n+m$. We simulate the data after setting $n=1000$ and $m=100$. 

GP with SE covariance kernel is still among the top performers (Table \ref{tab:sim1}). This simulation has a nonlinear periodic decision boundary; therefore, the performance of random forest and parametric models, including logistic regression and its regularized, deteriorate significantly. While all the kernel-based methods are suited for modeling the non-linear periodic decision boundary in \eqref{eq:sim31}, the SE kernel is better tuned than spectrum and boundrange kernels for modeling the hierarchical structure of ICD codes. This implies that all kernel-based methods perform well, but the GP with SE kernel performs better than SVM and KRR with spectrum and boundrange kernels.

\emph{\textul{Summary.}} Our simulation studies suggest that the kernel-based methods are better suited for applications in practice where we expect interactions among the diagnoses. Furthermore, GP with the SE kernel is easily extended to account for any biomedical information, is tuned for modeling the structure of ICD codes, and produces similar results as SVM or KRR with spectrum and boundrange kernels in the absence of any additional structure. The distinguishing feature of the proposed method is that the posterior MCMC draws can be used for quantifying uncertainty in predictions and parameter estimates, which is key in biomedical applications; therefore, we conclude that the GP with SE covariance kernel is better than other kernel based methods in quantifying uncertainty and in modeling the effect of chronic conditions on the response.

\begin{table}[ht]
  \caption{Performance comparisons for the three simulation studies. Every entry in the table is the average of its values across ten replications. SE-GP is the proposed method and SE kernel-based features are obtained from the kernel matrix estimated using the proposed method. If a method fails to produce results, then we indicate this using `-'.}
  \label{tab:sim1}  
  \centering
  {\tiny
    \begin{tabular}{|r|c|c|c|c|}
    \hline      
    \multicolumn{5}{|c|}{\textbf{FIRST SIMULATION}}  \\
    \hline
    & \textbf{AUC} & \textbf{Accuracy} & \textbf{Sensitivity} & \textbf{Specificity} \\ 
    \hline
    SE-GP & 0.96 & 0.96 & 0.96 & 0.96\\
    \hline
    & \multicolumn{4}{c|}{Logistic Regression with SE Kernel-Based Covariates} \\
    \hline 
    No Penalty & 0.95 & 0.95 & 0.96 & 0.95 \\
    Ridge Penalty & 0.95 & 0.96 & 0.96 & 0.95 \\
    Lasso Penalty & 0.95 & 0.95 & 0.96 & 0.95 \\
    \hline    
    & \multicolumn{4}{c|}{Logistic Regression} \\
    \hline 
    No Penalty & 0.96 & 0.96 & 0.96 & 0.96 \\ 
    Ridge Penalty & 0.96 & 0.96 & 0.96 & 0.96 \\ 
    Lasso Penalty & 0.96 & 0.96 & 0.96 & 0.96 \\ 
    \hline  
    & \multicolumn{4}{c|}{Kernel Ridge Regression}  \\
    \hline 
    Spectrum & 0.95 & 0.96 & 0.96 & 0.95 \\ 
    Boundrange & 0.95 & 0.95 & 0.96 & 0.95 \\ 
    \hline
    & \multicolumn{4}{c|}{Support Vector Machine}  \\
    \hline 
    Spectrum & 0.95 & 0.96 & 0.96 & 0.96 \\ 
    Boundrange & 0.95 & 0.95 & 0.96 & 0.95 \\
    \hline 
    & \multicolumn{4}{c|}{Random Forest}  \\
    \hline
    Default & 0.50 & 0.50 & 0.99 & 0.01 \\
    SE Kernel-Based Covariates & 0.95 & 0.95 & 0.96 & 0.95 \\ 
    \hline      
    \multicolumn{5}{|c|}{\textbf{SECOND SIMULATION}}  \\
    \hline
    & \textbf{AUC} & \textbf{Accuracy} & \textbf{Sensitivity} & \textbf{Specificity} \\ 
    \hline
    SE-GP & 1.00 & 0.99 & 1.00 & 0.98 \\ 
    \hline
    & \multicolumn{4}{c|}{Logistic Regression with SE Kernel-Based Covariates} \\
    \hline 
    No Penalty & 0.90 & 0.82 & 0.77 & 0.86 \\ 
    Ridge Penalty & 0.90 & 0.82 & 0.81 & 0.84 \\ 
    Lasso Penalty & 0.90 & 0.82 & 0.75 & 0.89 \\
    \hline    
    & \multicolumn{4}{c|}{Logistic Regression} \\
    \hline 
    No Penalty & 0.90 & 0.82 & 0.78 & 0.86 \\ 
    Ridge Penalty & 0.90 & 0.82 & 0.80 & 0.84 \\ 
    Lasso Penalty & 0.90 & 0.83 & 0.77 & 0.87 \\ 
    \hline  
    & \multicolumn{4}{c|}{Kernel Ridge Regression}  \\
    \hline 
    Spectrum & 0.92 & 0.86 & 0.86 & 0.86 \\ 
    Boundrange & 0.91 & 0.84 & 0.83 & 0.85 \\ 
    \hline
    & \multicolumn{4}{c|}{Support Vector Machine}  \\
    \hline 
    Spectrum & 0.92 & 0.86 & 0.86 & 0.86 \\ 
    Boundrange & 0.91 & 0.84 & 0.83 & 0.85 \\
    \hline        
    & \multicolumn{4}{c|}{Random Forest}  \\
    \hline
    Default & 0.50 & 0.45 & 1.00 & 0.00 \\
    SE Kernel-Based Covariates & 0.88 & 0.88 & 0.85 & 0.92 \\ 
    \hline
    \multicolumn{5}{|c|}{\textbf{THIRD SIMULATION}}  \\
    \hline
    & \textbf{AUC} & \textbf{Accuracy} & \textbf{Sensitivity} & \textbf{Specificity} \\ 
    \hline
    SE-GP & 0.94 & 0.91 & 0.90 & 0.92 \\ 
    \hline
    & \multicolumn{4}{c|}{Logistic Regression with SE Kernel-Based Covariates} \\
      \hline
    No Penalty & - & - & - & - \\             
    Ridge Penalty & 0.76 & 0.73 & 0.66 & 0.82 \\   
    Lasso Penalty & 0.75 & 0.73 & 0.67 & 0.81 \\ 
    \hline    
    & \multicolumn{4}{c|}{Logistic Regression} \\
      \hline
    No Penalty & - & - & - & - \\       
    Ridge Penalty & 0.89 & 0.84 & 0.85 & 0.84 \\  
    Lasso Penalty & 0.90 & 0.86 & 0.85 & 0.88 \\ 
    \hline  
    & \multicolumn{4}{c|}{Kernel Ridge Regression}  \\
    \hline 
    Spectrum & 0.84 & 0.80 & 0.82 & 0.78 \\ 
    Boundrange & 0.81 & 0.78 & 0.81 & 0.74 \\  
    \hline
    & \multicolumn{4}{c|}{Support Vector Machine}  \\
    \hline 
    Spectrum & 0.85 & 0.81 & 0.82 & 0.81 \\   
    Boundrange & 0.83 & 0.79 & 0.77 & 0.81 \\ 
    \hline
    & \multicolumn{4}{c|}{Random Forest}  \\
    \hline
    Default & 0.85 & 0.86 & 0.95 & 0.76 \\
    SE Kernel-Based Covariates & 0.89 & 0.89 & 0.91 & 0.88 \\
    \hline    
  \end{tabular}    
  }%
\end{table}

\subsection{Real data analysis}
\label{sec:real-data-analysis}

The UIHC EHR data contains information about 1660 patients, including their diagnoses, marital status,  and the primary cancer site. There are six types of cancer sites: brain and other nervous system (brain), breast, urinary system, respiratory system, female genital system, and digestive system. \citet{Caletal17} define 58 chronic conditions using ICD codes that do not include any diagnosis related to neoplasms. For the feature map in \eqref{eq:thm-feat} that defines the SE kernel, we set $C=58$, $\Rb = \Ib$, and $w_c \propto \kappa_1(\Tb_c, \Tb_c)$, where $\Tb_c$ is the set of ICD codes definining $c$th chronic condition. We include marital status as the only demographic predictor, so $p=2$ in \eqref{eq:log1}. The major and minor goals of the analysis are to estimate the marginal associations of the 58 chronic conditions with the cancer sites and to predict the cancer primary site using the patient diagnoses and marital status, respectively.

We use the methods from the previous section to achieve both goals. Five classification models are used for predicting the six primary sites using $y = 1$ to denote the presence of cancer in brain, breast, urinary system, respiratory system, and female genital system, respectively, in the five models. For every model, all the methods are trained on 80\% of the full data, their performance is evaluated on the excluded data, and the setup is replicated ten times.
There is no principled way of using chronic conditions as covariates in logistic regression and its penalized extension, so we only use KRR, SVM, and SE-GP for estimating the marginal associations between the 58 chronic conditions and primary cancer site. To this end, we include $\Tb_1, \ldots, \Tb_{58}$ as the diagnoses of 58 pseudo patients in the test data and summarize the marginal associations using the estimate of Pr$(y=1 \mid \Tb_c)$, where $\Tb_c$ is the diagnosis for the $c$th pseudo patient ($c=1, \ldots, 58$). Unlike SVM and KRR, SE-GP also provides 95\% credible intervals for the marginal associations using the MCMC draws of Pr$(y=1)$.

SE-GP is among the top performers in all the five classification models (Tables \ref{tab:r1}--\ref{tab:r5}). We do not provide results for KRR and logistic regression with no penalty because the former fails due to a line search error and the latter cannot be used because the number of dummy variables is larger than the sample size. \citet{Caletal17} have developed a comprehensive list of chronic conditions after a careful study, so we expect a subset of them to be related to primary cancer sites. This is confirmed by a relatively high accuracy of logistic regression with the ridge and lasso penalties for predicting the cancer in brain and urinary system; however, these methods do not use the information encoded in ICD codes, so their sensitivity and specificity vary a lot; for example, logistic regression with the lasso penalty has a high accuracy of 0.90 for predicting the cancer in urinary system, but its specificity in the same model is 0.36.

On the other hand, SE-GP is among the top performers in terms of accuracy, sensitivity, and specificity in all the five classification models. This also implies that the AUC of SE-GP is the highest among all methods in all the five classification models. The features obtained from the SE kernel matrix are also promising in that if we use them as covariates in logistic regression and its penalized extensions, then the AUCs are relatively large in all the five classification models. The performance of default random forest is the worst in all five classification models. AUC and accuracy increase after including SE kernel-based covariates, but the results are still poorer than all other methods.

SE-GP also performs better than SVM with spectrum and boundrange kernels because SE kernel accounts for the additional structure provided by the 58 chronic conditions. The SE-GP is a restriction of the boundrange kernel that is tuned for EHR data analysis, so its sensitivity and specificity are much higher than those of the SVM in all the five classification models. Most importantly, the estimates of marginal associations between chronic conditions and primary cancer sites agree closely with those obtained using the SVM with spectrum and boundrange kernels  (Table \ref{tab:c}). A key feature of the SE-GP approach is that it provides 95\% credible intervals in addition to the point estimates, which are important for characterizing uncertainty in biomedical applications.  Additionally, these credible intervals cover the corresponding estimates obtained using SVM. Based on our results in simulations and this section, we conclude that SE-GP outperforms its competitors in estimating the marginal associations among chronic conditions and primary cancer sites and in predicting the cancer primary site using diagnoses and demographic information as predictors.

\begin{table}[ht]
  \caption{Performance in the classification model with $y=1$ denoting the presence of cancer in the brain and other nervous system. Every entry in the table is the average of its values across ten replications. }
  \label{tab:r1}
  \centering
  {\tiny
  \begin{tabular}{|r|c|c|c|c|}
    \hline
    & \textbf{AUC} & \textbf{Accuracy} & \textbf{Sensitivity} & \textbf{Specificity} \\ 
    \hline
    SE-GP & 0.93 & 0.92 & 0.86 & 0.93 \\
    \hline
    & \multicolumn{4}{c|}{Logistic Regression with SE Kernel-Based Covariates} \\
    \hline 
    No Penalty & 0.91 & 0.88 & 0.82 & 0.89 \\ 
    Lasso Penalty & 0.91 & 0.89 & 0.81 & 0.91 \\ 
    Ridge Penalty &0.90 & 0.90 & 0.81 & 0.91 \\ 
    \hline
    & \multicolumn{4}{c|}{Logistic Regression} \\
    \hline 
    Lasso Penalty & 0.90 & 0.92 & 0.81 & 0.94  \\ 
    Ridge Penalty & 0.93 & 0.90 & 0.89 & 0.90 \\
    \hline  
    & \multicolumn{4}{c|}{Support Vector Machine}  \\
    \hline 
    Boundrange & 0.94 & 0.91 & 0.88 & 0.92 \\ 
    Spectrum & 0.94 & 0.90 & 0.90 & 0.90 \\
    \hline    
    & \multicolumn{4}{c|}{Random Forest}  \\
    \hline
    Default & 0.50 & 0.11 & 1.00 & 0.00 \\ 
    SE Kernel-Based Covariates & 0.72 & 0.92 & 0.45 & 0.98 \\ 
    \hline
  \end{tabular}    
  }%
\end{table}

\begin{table}[ht]
  \caption{Performance in the classification model with $y=1$ denoting the presence of cancer in the breast. Every entry in the table is the average of its values across ten replications. }
  \label{tab:r2}
  \centering
  {\tiny
  \begin{tabular}{|r|c|c|c|c|}
    \hline
    & \textbf{AUC} & \textbf{Accuracy} & \textbf{Sensitivity} & \textbf{Specificity} \\ 
    \hline
    SE-GP & 0.77 & 0.72 & 0.75 & 0.71 \\ 
    \hline
    & \multicolumn{4}{c|}{Logistic Regression with SE Kernel-Based Covariates} \\
    \hline 
    No Penalty & 0.69 & 0.54 & 0.80 & 0.51 \\ 
    Lasso Penalty & 0.68 & 0.59 & 0.74 & 0.56 \\ 
    Ridge Penalty &0.68 & 0.58 & 0.74 & 0.56 \\ 
    \hline
    & \multicolumn{4}{c|}{Logistic Regression} \\
    \hline 
    Lasso Penalty & 0.74 & 0.66 & 0.74 & 0.66 \\ 
    Ridge Penalty & 0.73 & 0.72 & 0.69 & 0.73 \\ 
    \hline  
    & \multicolumn{4}{c|}{Support Vector Machine}  \\
    \hline 
    Boundrange & 0.68 & 0.67 & 0.63 & 0.68 \\ 
    Spectrum & 0.74 & 0.67 & 0.75 & 0.66 \\
    \hline
    & \multicolumn{4}{c|}{Random Forest}  \\
    \hline
    Default & 0.50 & 0.11 & 1.00 & 0.00 \\
    SE Kernel-Based Covariates & 0.53 & 0.80 & 0.18 & 0.88 \\
    \hline
  \end{tabular}    
  }%
\end{table}

\begin{table}[ht]
  \caption{Performance in the classification model with $y=1$ denoting the presence of cancer in the urinary system. Every entry in the table is the average of its values across ten replications. }
  \label{tab:r3}
  \centering
  {\tiny
  \begin{tabular}{|r|c|c|c|c|}
    \hline
    & \textbf{AUC} & \textbf{Accuracy} & \textbf{Sensitivity} & \textbf{Specificity} \\ 
    \hline
    SE-GP & 0.80 & 0.78 & 0.71 & 0.79\\
    \hline
    & \multicolumn{4}{c|}{Logistic Regression with SE Kernel-Based Covariates} \\
    \hline 
    No Penalty & 0.75 & 0.68 & 0.71 & 0.68 \\ 
    Lasso Penalty& 0.75 & 0.70 & 0.70 & 0.70 \\ 
    Ridge Penalty& 0.75 & 0.75 & 0.63 & 0.77 \\ 
    \hline
    & \multicolumn{4}{c|}{Logistic Regression} \\
    \hline 
    Lasso Penalty & 0.69 & 0.89 & 0.36 & 0.97 \\ 
    Ridge Penalty & 0.74 & 0.76 & 0.68 & 0.77 \\
    \hline  
    & \multicolumn{4}{c|}{Support Vector Machine}  \\
    \hline 
    Boundrange & 0.60 & 0.71 & 0.43 & 0.75 \\ 
    Spectrum & 0.68 & 0.76 & 0.53 & 0.79 \\
    \hline
    & \multicolumn{4}{c|}{Random Forest}  \\
    \hline
    Default & 0.50 & 0.12 & 1.00 & 0.00 \\ 
    SE Kernel-Based Covariates & 0.58 & 0.89 & 0.17 & 0.98 \\ 
    \hline    
  \end{tabular}    
  }%
\end{table}

\begin{table}[ht]
  \caption{Performance in the classification model with $y=1$ denoting the presence of cancer in the respiratory system. Every entry in the table is the average of its values across ten replications. }
  \label{tab:r4}
  \centering
  {\tiny
  \begin{tabular}{|r|c|c|c|c|}
    \hline
    & \textbf{AUC} & \textbf{Accuracy} & \textbf{Sensitivity} & \textbf{Specificity} \\ 
    \hline
    SE-GP & 0.83 & 0.75 & 0.81 & 0.73 \\ 
    \hline
    & \multicolumn{4}{c|}{Logistic Regression with SE Kernel-Based Covariates} \\
    \hline
    No Penalty & 0.67 & 0.56 & 0.78 & 0.52 \\ 
    Lasso Penalty & 0.67 & 0.57 & 0.76 & 0.53 \\ 
    Ridge Penalty & 0.67 & 0.56 & 0.78 & 0.52 \\
    \hline
    & \multicolumn{4}{c|}{Logistic Regression} \\
    \hline 
    Lasso Penalty & 0.81 & 0.74 & 0.73 & 0.74 \\ 
    Ridge Penalty & 0.83 & 0.75 & 0.83 & 0.74 \\ 
    \hline  
    & \multicolumn{4}{c|}{Support Vector Machine}  \\
    \hline 
    Boundrange & 0.77 & 0.70 & 0.73 & 0.70 \\ 
    Spectrum & 0.82 & 0.75 & 0.78 & 0.74 \\
    \hline
    & \multicolumn{4}{c|}{Random Forest}  \\
    \hline
    Default & 0.50 & 0.16 & 1.00 & 0.00 \\ 
    SE Kernel-Based Covariates & 0.55 & 0.84 & 0.13 & 0.97 \\  
    \hline
  \end{tabular}    
  }%
\end{table}

\begin{table}[ht]
  \caption{Performance in the classification model with $y=1$ denoting the presence of cancer in female gential system. Every entry in the table is the average of its values across ten replications.  }
  \label{tab:r5}
  \centering
  {\tiny
  \begin{tabular}{|r|c|c|c|c|}
    \hline
    & \textbf{AUC} & \textbf{Accuracy} & \textbf{Sensitivity} & \textbf{Specificity} \\ 
    \hline
    SE-GP & 0.86 & 0.83 & 0.72 & 0.86 \\  
    \hline
    & \multicolumn{4}{c|}{Logistic Regression with SE Kernel-Based Covariates} \\
    \hline
    No Penalty & 0.78 & 0.73 & 0.66 & 0.75 \\ 
    Lasso Penalty &  0.78 & 0.75 & 0.65 & 0.78 \\ 
    Ridge Penalty & 0.78 & 0.75 & 0.66 & 0.78 \\
    \hline
    & \multicolumn{4}{c|}{Logistic Regression} \\
    \hline 
    Lasso Penalty & 0.84 & 0.77 & 0.75 & 0.77 \\ 
    Ridge Penalty & 0.85 & 0.81 & 0.77 & 0.82 \\ 
    \hline  
    & \multicolumn{4}{c|}{Support Vector Machine}  \\
    \hline 
    Boundrange & 0.77 & 0.75 & 0.68 & 0.77 \\ 
    Spectrum & 0.83 & 0.80 & 0.70 & 0.83 \\
    \hline
    & \multicolumn{4}{c|}{Random Forest}  \\
    \hline
    Default & 0.50 & 0.21 & 1.00 & 0.00 \\ 
    SE Kernel-Based Covariates & 0.70 & 0.84 & 0.45 & 0.95 \\   
    \hline
  \end{tabular}    
  }%
\end{table}
  
\begin{table}[ht]
  \caption{Summary of the marginal associations between 58 chronic conditions and the six primary cancer sites. The marginal association between a chronic condition and a primary cancer site denoted using $y=1$ is defined as the point estimate of Pr$(y=1 \mid \Tb_c)$, where $\Tb_c$ is the diagnosis of the pseudo patient representing the chronic condition $c$. For SE-GP, we define the point estimate of Pr$(y=1  \mid x^*)$ to be its posterior mean and use the 95\% credible interval to provide an estimate of uncertainty. For the 58 pseudo patients and six primary cancer sites, we obtain the point estimate and 95\% credible interval using the MCMC draws. The table includes only those chronic conditions for which the point estimates of the marginal associations are greater than 0.5 and a `-' indicates the absence of such chronic conditions.}
  \label{tab:c}  
  \centering
  {\tiny
  \begin{tabular}{|r|c|c|}
    \hline
    \textbf{Kernel} & \textbf{Chronic Condition} & \textbf{Estimate} \\ 
    \hline
    \multicolumn{3}{c|}{\textbf{Brain and Other Nervous System}} \\
    \hline
    SE-GP & OTHER NEUROLOGICAL DISEASES & 0.8075, (0.5199, 0.9649) \\ 
    Boundrange & OTHER NEUROLOGICAL DISEASES & 0.8057 \\ 
    Spectrum & OTHER NEUROLOGICAL DISEASES & 0.6570 \\
    \hline
    \multicolumn{3}{c|}{\textbf{Breast}} \\
    \hline
    SE-GP & DEPRESSION AND MOOD DISEASES & 0.6909, (0.2960, 0.9451) \\
    Boundrange & - & - \\ 
    Spectrum & OTHER CARDIOVASCULAR DISEASES & 0.5799 \\ 
    Spectrum & OSTEOPOROSIS & 0.5884 \\ 
    Spectrum & ALLERGY  & 0.6005 \\ 
    Spectrum & DEPRESSION AND MOOD DISEASES & 0.6454 \\ 
    Spectrum & VENOUS AND LYMPHATIC DISEASES & 0.7290 \\     
    \hline
    \multicolumn{3}{c|}{\textbf{Urinary System}} \\
    \hline
    SE-GP & - & -\\ 
    Boundrange & - & - \\ 
    Spectrum & - & - \\    
    \hline
    \multicolumn{3}{c|}{\textbf{Respiratory System}} \\
    \hline
    SE-GP & OTHER RESPIRATORY DISEASES & 0.7973, (0.5421, 0.9483) \\ 
    SE-GP & COPD, EMPHYSEMA, CHRONIC BRONCHITIS & 0.8815, (0.6181, 0.9880) \\ 
    Boundrange & OTHER RESPIRATORY DISEASES & 0.6901 \\ 
    Boundrange & COPD, EMPHYSEMA, CHRONIC BRONCHITIS & 0.7575 \\ 
    Spectrum & OTHER RESPIRATORY DISEASES & 0.5024 \\ 
    Spectrum & OTHER PSYCHIATRIC AND BEHAVIORAL DISEASES & 0.5871 \\ 
    Spectrum & COPD, EMPHYSEMA, CHRONIC BRONCHITIS & 0.8992 \\ 
    \hline
    \multicolumn{3}{c|}{\textbf{Female Genital System}} \\
    \hline
    SE-GP & OBESITY & 0.6616, (0.3507, 0.8915) \\ 
    Boundrange & OBESITY & 0.5276 \\
    Spectrum & - & - \\ 
    \hline
    \multicolumn{3}{c|}{\textbf{Digestive System}} \\
    \hline    
    SE-GP & COLITIS AND RELATED DISEASES & 0.5114, (0.1576, 0.8618) \\ 
    SE-GP & OTHER METABOLIC DISEASES & 0.5382, (0.2198, 0.8286) \\ 
    SE-GP & INFLAMMATORY BOWEL DISEASES & 0.5874, (0.2478, 0.8788) \\ 
    SE-GP & CHRONIC PANCREAS, BILIARY TRACT AND GALLBLADDER DISEASES & 0.7002, (0.3003, 0.9464) \\ 
    SE-GP & CHRONIC LIVER DISEASES &0.7893,  (0.4196, 0.9686) \\ 
    SE-GP & ESOPHAGUS, STOMACH AND DUODENUM DISEASES & 0.8413, (0.5631, 0.9754) \\ 
    Boundrange & OTHER METABOLIC DISEASES & 0.5736 \\ 
    Boundrange & INFLAMMATORY BOWEL DISEASES & 0.8326 \\ 
    Boundrange & OTHER DIGESTIVE DISEASES & 0.9170 \\ 
    Boundrange & CHRONIC PANCREAS, BILIARY TRACT AND GALLBLADDER DISEASES & 0.9281 \\ 
    Boundrange & COLITIS AND RELATED DISEASES & 0.9389 \\ 
    Boundrange & ESOPHAGUS, STOMACH AND DUODENUM DISEASES & 0.9949 \\ 
    Boundrange & CHRONIC LIVER DISEASES & 0.9958 \\ 
    Spectrum & OTHER METABOLIC DISEASES & 0.5052 \\ 
    Spectrum & OTHER DIGESTIVE DISEASES & 0.6525 \\ 
    Spectrum & COLITIS AND RELATED DISEASES & 0.6580 \\ 
    Spectrum & CHRONIC PANCREAS, BILIARY TRACT AND GALLBLADDER DISEASES & 0.7498 \\ 
    Spectrum & CHRONIC LIVER DISEASES & 0.8930 \\ 
    Spectrum & ESOPHAGUS, STOMACH AND DUODENUM DISEASES & 0.9323 \\     
    \hline
  \end{tabular}    
  }%
\end{table}

\section{Discussion}
\label{s:discuss}

Our approach can be extended in several ways. First, the sample size in the motivating application is small so that repeated computations and evaluation of the kernel function is relatively inexpensive. This, however, becomes problematic as the sample size becomes moderately large. The kernels developed in this work can be immediately extended for nonparametric regression and classification in large sample settings using low-rank kernels based on inducing points \citep{QuiRas05}. Second, Every element of the kernel matrix is a similarity measure between a pair of patients. This can be used an input in an algorithm for unsupervised learning using ICD codes as inputs, such as clustering of subsets, data visualization, and dimension reduction \citep[Chapters 6, 8]{ShaCri04}. The PG-DA strategy outlined in Section \ref{mdl-class} is easily extended to model zero-inflated negative binomial responses following \citet{Nee19}. Finally, We are exploring application of product partition model with regression on diagnoses, where the similarity measure on diagnoses is defined using the kernel matrix. 

\section*{Acknowledgement}

We thank UIHC for sharing the EHR data. Sanvesh Srivastava's research is supported by grants from the Office of Naval Research (ONR-BAA N000141812741) and the National Science Foundation (DMS-1854667/1854662). Stephanie Gilbertson-White's research is supported by the University of Iowa, College of Nursing Center on Advancing Multimorbidity Science (CAMS) funded be the National Institute for Nursing Research (P20NR018081). The R code for simulations is available at \url{https://github.com/blayes/ehr-data-analysis}.

\bibliographystyle{Chicago}
\bibliography{papers}

\newpage

\pagenumbering{gobble}
\pagenumbering{arabic}
\renewcommand*{\thesection}{\Alph{section}}
\setcounter{section}{0}

\section*{Supplemental Document for Gaussian Process Regression and Classification using International Classification of Disease Codes as Covariates}

\section{Proof of Theorem 1}
\label{asec:proof-theorem-1}

We recall the definitions of kernel functions on ICD codes. Let $\psib(t)$ be the feature map of an ICD code $t \in \Tcal$ and $\Fcal = \{\psib(t): t \in \Tcal\}$ be the feature space. For every $t \in \Tcal$, $\psib(t) \in \{0, 1\}^{|t^*||\Bcal|}$ and is defined as 
\begin{align}
  \label{aeq:1}
  \psib: \Tcal \mapsto \Fcal, \quad \psib(t) =  \{\psib_b(t)\}_{b \in \Bcal}, \quad \psib_b(t) = \{1(b = t_{1:1}), \ldots, 1(b = t_{1:|t^*|})\}^\T, 
\end{align}
where $1(\cdot)$ is 1 if $\cdot$ is true and 0 otherwise. The kernel function on $\Tcal \times \Tcal$ is defined using $\psib(\cdot)$ and  known $|t^*| \times |t^*|$ symmetric positive semi-definite matrices $(\Lambdab_b)_{b \in \Bcal}$ as 
\begin{align}
  \label{aeq:e1}
  \kappa_0(t, t' \mid \Lambdab) = \underset{b \in \Bcal} {\sum} \psib^\T_b(t) \Lambdab_b \psib_b(t') \equiv \psib^\T(t) \Lambdab \psib(t'), \; t, t' \in \Tcal, \quad \Lambdab = \diag\{\Lambdab_b: b \in \Bcal\},
\end{align}
where $\Lambdab$ is a block diagonal matrix.

We use $\kappa_0$ in \eqref{aeq:e1} to define the kernel function on subsets of ICD codes that represent diagnoses. Let $2^{\Tcal}$ be the power set of $\Tcal$, $\tb = \{t_1, \ldots, t_r \}$ be an element of $2^{\Tcal}$, $\Psi(\tb)$ be the feature map of $\tb$, and $\Fcal_\Psi = \{\Psib(\tb): \tb \in 2^{\Tcal}\}$ be the feature space of subsets of ICD codes. The feature vector $\Psib(\tb)$ and the kernel function $\kappa_1$ on $2^{\Tcal} \times 2^{\Tcal}$ are defined using $\psib(\cdot)$  in \eqref{aeq:1} and $\kappa_0(\cdot, \cdot \mid \Lambdab)$ in \eqref{aeq:e1}, respectively, as
\begin{align}
  \label{aeq:3}
  \Psib(\tb) = \sum_{{t \in \tb}} \psib(t) , \quad \kappa_1(\tb, \tb' \mid \Lambdab ) = \sum_{{t \in \tb}} \sum_{{t' \in \tb'}} \kappa_0(t, t' \mid \Lambdab) = \Psib^\T(\tb) \Lambdab \Psib(\tb') , \quad \tb, \tb' \in 2^{\Tcal},
\end{align}
where $\Psib(\tb)$ is a $(|t^*||\Bcal|)$-dimensional vector with entries in $\{0\} \cup \NN $ and $\Lambdab$ is defined as in \eqref{aeq:e1}.

We account the additional structure provided by the $C$ chronic conditions. A patient's diagnoses are represented using the set $\Tb = \{\tb_1, \ldots, \tb_c, \ldots, \tb_C\} \in \left( 2^{\Tcal} \right)^C $, where $\tb_c \subseteq \Tb_c$. A kernel function on $\left( 2^{\Tcal} \right)^C \times \left( 2^{\Tcal} \right)^C$ that accounts for the structure imposed by $C$ chronic conditions is
\begin{align}
  \label{aeq:4}
  \kappa_2(\Tb, \Tb' \mid \wb, \Lambdab) = \sum_{c=1}^C w_c \, \kappa_1(\tb_c, \tb'_c \mid \Lambdab), \quad \Tb, \Tb' \in \left( 2^{\Tcal} \right)^C,  \quad w_1, \ldots, w_C > 0, 
\end{align}
where $\Lambdab$ is defined in \eqref{aeq:e1}, $\wb = (w_1, \ldots, w_C)$, and $w_c$ indicates the importance of chronic condition $c$ in defining the similarity between two patients with diagnoses $\Tb$ and $\Tb'$, respectively. Finally, we normalize $\kappa_2$ to obtain our kernel function 
\begin{align}
  \label{aeq:7}
  \kappa (\Tb, \Tb' \mid \wb, \Lambdab) &= \frac{\kappa_2 (\Tb, \Tb' \mid \wb, \Lambdab)}{\left\{ \kappa_2 (\Tb, \Tb \mid \wb, \Lambdab)\right\}^{1/2} \left\{ \kappa_2 (\Tb', \Tb' \mid \wb, \Lambdab) \right\}^{1/2}}, \quad \Tb, \Tb' \in \left( 2^{\Tcal} \right)^C.
\end{align}

We restate Theorem 1. 
\begin{theorem}\label{athm1}
  Let $\{\Lambdab_b\}_{b \in \Bcal}$ and $\wb$ be defined as in \eqref{aeq:e1} and \eqref{aeq:4}, $\Tb, \Tb' \in (2^{\Tcal})^C$, $\sigma^2 > 0$, $\phi > 0 $, and 
  \begin{align}
    \label{eq:5}
    d(\Tb, \Tb' \mid \wb, \Lambdab) &= \left\{ \kappa(\Tb, \Tb \mid \wb, \Lambdab) + \kappa(\Tb', \Tb' \mid \wb, \Lambdab) - 2 \kappa(\Tb, \Tb' \mid \wb, \Lambdab) \right\}^{1/2}. 
  \end{align}
  Then,
  \begin{enumerate}
  \item given $\Lambdab$ and $\wb$, $\kappa(\cdot, \cdot \mid \wb, \Lambdab)$ in \eqref{aeq:7} and $d(\cdot, \cdot \mid \wb, \Lambdab)$ in \eqref{eq:5} are valid kernel and distance functions on $(2^{\Tcal})^C \times (2^{\Tcal})^C$; 
  \item if $\Lambdab = \Rb^{\top} \Rb$, where $\Rb$ is an upper triangular matrix, then the feature map of $\kappa$ is 
    \begin{align}
      \label{aeq:thm-feat}
      \overline \Psib_{\wb, \Rb} (\Tb) = \frac{\Psib_{\wb, \Rb} (\Tb)}{\| \Psib_{\wb, \Rb} (\Tb) \|}, \quad \Psib_{\wb, \Rb} (\Tb) = \{w_1^{1/2} \Rb \Psib(\tb_1), \ldots, w_C^{1/2} \Rb \Psib(\tb_C)\}^\T,
    \end{align}
    where $\Psib_{\wb, \Rb} (\Tb)$ is a $(C|t^*||\Bcal|)$-dimensional vector; and 
  \item the equivalents of polynomial kernel of order $s$ and  $\gamma$-exponential kernel on $(2^{\Tcal})^C \times (2^{\Tcal})^C$, respectively, are defined as
  \end{enumerate}  
  \begin{align}
    \label{aeq:thm-kern}
    \kappa_s(\Tb, \Tb' \mid \sigma^2, \wb, \Lambdab) &=  \sigma^2\left\{1 + \kappa(\Tb, \Tb' \mid \wb, \Lambdab) \right\}^s, \quad  s \in \NN,
                                                      \nonumber \\
    \kappa_\gamma^e(\Tb, \Tb'  \mid \sigma^2, \phi, \wb, \Lambdab) &= \sigma^2 e^{- \phi \, \{d(\Tb, \Tb' \mid  \wb, \Lambdab)\}^\gamma}, \quad \gamma \in [1, 2],
  \end{align}
  where $\sigma^2$ and $\phi$  play the role of variance and inverse length-scale parameters, respectively, and $\kappa_\gamma^e$ reduces to the exponential and SE kernels when $\gamma $ equals 1 and 2, respectively. 
\end{theorem}
\begin{proof}  
  \underline{\textbf{Proof of (1):}} 
  Because $\Lambdab_b$ is a symmetric positive semi-definite matrix, then Cholesky decomposition implies that there exists an upper triangular matrix with a non-negative diagonal elements $\Rb_b$ such that $\Rb_b^\top \Rb_b = \Lambdab_b$. Let $\Rb = \diag\{\Rb_b: b \in \Bcal\}$ and $\Lambdab = \diag\{\Lambdab_b: b \in \Bcal\}$. Then, $\Rb^\top \Rb = \Lambdab$ and \eqref{aeq:e1} implies that
  \begin{align}
    \label{eq:pf-athm1-1}
    \kappa_0(t, t' \mid \Lambdab) = \psib(t)^\top \Rb^\top \Rb  \psib(t') = \left\{ \Rb  \psib(t) \right\}^\top \Rb  \psib(t') = \left\{ \Rb  \psib(t') \right\}^\top \Rb  \psib(t)  = \kappa_0(t', t \mid \Lambdab).
  \end{align}
  For any $t_1, \ldots, t_n \in \Tcal$, define $\Gb$ to be the kernel matrix with $\Gb_{ij} = \kappa_0(t_i, t_j)$ ($i, j \in \{1, \ldots, n\}$). Then, $\Gb$ is symmetric because $\Gb_{ij} = \kappa_0(t_i, t_j) = \kappa_0( t_j, t_i) =\Gb_{ji} $ using \eqref{eq:pf-athm1-1} and, for any $\vb \in \RR^n$,
  \begin{align}
    \label{eq:pf-athm1-2}
    \vb^\top \Gb \vb &= \sum_{i,j=1}^n v_i \Gb_{ij} v_j = \sum_{i=1}^n \sum_{j=1}^n v_i \kappa_0(t_i, t_j) v_j \overset{(i)}{=} \sum_{i=1}^n \sum_{j=1}^n v_i \left\{ \Rb  \psib(t_i) \right\}^\top \Rb  \psib(t_j) v_j \nonumber\\
    &=   \left\{ \Rb  \sum_{i=1}^n v_i \psib(t_i) \right\}^\top \Rb \sum_{j=1}^n  v_j   \psib(t_j) = \| \Rb  \sum_{i=1}^n v_i \psib(t_i) \|^2 \geq 0,
  \end{align}
  where $(i)$ follows from \eqref{eq:pf-athm1-1}; therefore,  $\kappa_0(\cdot, \cdot \mid \Lambdab)$ is a kernel function with feature map $\Rb  \psib(t)$ for any $t \in \Tcal$. Similarly, \eqref{aeq:3} implies that  for any $\tb, \tb' \in 2^{\Tcal}$,
  \begin{align}
    \label{eq:pf-athm1-3}
    \kappa_1(\tb, \tb' \mid \Lambdab) &=  \left\{ \Rb \sum_{t \in \tb} \psib(t) \right\}^\top \Rb \sum_{t' \in \tb'}  \psib(t') =
    \left\{ \Rb \Psib(\tb) \right\}^\top \Rb \Psib(\tb') \nonumber\\
    &=  \left\{ \Rb \Psib(\tb') \right\}^\top \Rb \Psib(\tb) =
    \left\{ \Rb \sum_{t' \in \tb'}  \psib(t') \right\}^\top \Rb \sum_{t \in \tb}  \psib(t)  .
  \end{align}
  Following the same arguments used in  \eqref{eq:pf-athm1-2}, we have that $\kappa_1(\cdot, \cdot \mid \Lambdab)$ is a kernel function with feature map $\Rb  \Psib(\tb)$ for any $\tb \in 2^{\Tcal}$. 
  
  The closure properties of kernel functions implies that $\kappa_2$ in \eqref{aeq:4} is a valid kernel function.
  Proposition 3.22 on Page 75 in \citet{ShaCri04} implies that $w_c \kappa_1(\cdot, \cdot \mid \Lambdab)$ is a kernel function because $\kappa_1(\cdot, \cdot \mid \Lambdab)$ is a kernel function and $w_c > 0$ ($c=1, \ldots, C$). The same proposition also implies that $\kappa_2(\Tb, \Tb' \mid \wb, \Lambdab)=\sum_{c=1}^C w_c \kappa_1(\tb_c, \tb_c \mid \Lambdab)$ is a valid kernel function for any $\Tb =\{\tb_1, \ldots, \tb_C\}, \Tb'=\{\tb_1', \ldots, \tb_C'\} \in (2^{\Tcal})^C$. The feature map of $\kappa_2(\cdot, \cdot \mid \wb, \Lambdab)$ is $\Psib_{\wb, \Rb}(\Tb)$ for any $\Tb \in (2^{\Tcal})^C$, and it is obtained following the same steps in \eqref{eq:pf-athm1-2} and   \eqref{eq:pf-athm1-3} as 
  \begin{align}\label{eq:pf-athm1-4}    
    \kappa_2(\Tb, \Tb' \mid \wb, \Lambdab) &= \sum_{c=1}^C w_c \kappa_1(\tb_c, \tb'_c \mid \Lambdab)  = \sum_{c=1}^C 
    \left\{ w_c^{1/2} \Rb  \Psib(\tb_c) \right\}^\top w_c^{1/2}  \Rb \Psib(\tb'_c) \nonumber\\
    &= \{w_1^{1/2}  \Rb \Psib(\tb_1), \ldots, w_C^{1/2} \Rb  \Psib(\tb_C)\}^\top
      \{w_1^{1/2}  \Rb \Psib(\tb_1'), \ldots, w_C^{1/2} \Rb \Psib(\tb_C')\} \nonumber\\
    &\equiv \{\Psib_{\wb, \Rb} (\Tb)\}^\top \Psib_{\wb, \Rb} (\Tb'). 
  \end{align}
  The dimension of $\Psib_{\wb, \Rb} (\Tb)$ is $C$  times the dimension of $\psib(t)$ for any $t \in \Tcal$, which equals $C|t^*||\Bcal|$. 
  
  Finally, the kernel $\kappa$ in \eqref{aeq:7} is obtained by the normalization of $\kappa_2$. First, normalize the feature map $\Psib_{\wb, \Rb} (\Tb)$ for any $\Tb \in (2^{\Tcal})^C$ by its Euclidean norm as 
  \begin{align}
    \label{eq:pf-athm1-51}
    \overline \Psib_{\wb, \Rb} (\Tb) &= \frac{\Psib_{\wb, \Rb} (\Tb)}{\| \Psib_{\wb, \Rb} (\Tb) \|} \overset{(i)}{\in} [0, 1]^{C|t^*||\Bcal|}, \nonumber\\  
    \| \Psib_{\wb, \Rb} (\Tb) \|^2 &= \{\Psib_{\wb, \Rb} (\Tb)\}^\top \Psib_{\wb, \Rb} (\Tb) = \kappa_2(\Tb, \Tb \mid \wb, \Lambdab), 
  \end{align}
  where $(i)$ follows from the Cauchy-Schwartz inequality.  Second, noting that $\| \Psib_{\wb, \Rb} (\Tb) \| = \{\kappa_2(\Tb, \Tb \mid \wb, \Lambdab)\}^{1/2}$, define the normalized $\kappa_2$ kernel with feature map $\overline \Psib_{\wb, \Rb} (\Tb) $ as
  \begin{align}
    \label{eq:pf-athm1-5}
    \kappa(\Tb, \Tb' \mid \wb, \Lambdab)    &= \{\overline \Psib_{\wb, \Rb} (\Tb) \}^\top \overline \Psib_{\wb, \Rb} (\Tb') = \frac{\{\Psib_{\wb, \Rb} (\Tb)\}^\top \Psib_{\wb, \Rb} (\Tb')}{\| \Psib_{\wb, \Rb} (\Tb) \| \| \Psib_{\wb, \Rb} (\Tb')
    \|} \nonumber\\
    &= \frac{\kappa_2(\Tb, \Tb' \mid \wb, \Lambdab)}{\{\kappa_2(\Tb, \Tb \mid \wb, \Lambdab)\}^{1/2} \{\kappa_2(\Tb', \Tb' \mid \wb, \Lambdab)\}^{1/2}}.
  \end{align}
  Using the feature map $\overline \Psib_{\wb, \Rb} (\Tb) $ in the definition of distance function in \eqref{eq:5}, we have that 
  \begin{align}
    \label{eq:eq:pf-athm1-6}
    \{d(\Tb, \Tb' \mid \wb, \Lambdab)\}^2 &= \kappa(\Tb, \Tb \mid \wb, \Lambdab) + \kappa(\Tb', \Tb' \mid \wb, \Lambdab) - 2 \kappa(\Tb, \Tb' \mid \wb, \Lambdab) \nonumber\\
    &= \{\overline \Psib_{\wb, \Rb} (\Tb) \}^\top \overline \Psib_{\wb, \Rb} (\Tb) + 
      \{\overline \Psib_{\wb, \Rb} (\Tb') \}^\top \overline \Psib_{\wb, \Rb} (\Tb') 
      - 2 \{\overline \Psib_{\wb, \Rb} (\Tb) \}^\top \overline \Psib_{\wb, \Rb} (\Tb') \nonumber\\
    &= \| \overline \Psib_{\wb, \Rb} (\Tb) \|^2 +  \| \overline \Psib_{\wb, \Rb} (\Tb') \|^2
      - 2 \{\overline \Psib_{\wb, \Rb} (\Tb) \}^\top \overline \Psib_{\wb, \Rb} (\Tb') \nonumber\\
    &= \| \overline \Psib_{\wb, \Rb} (\Tb) -  \overline \Psib_{\wb, \Rb} (\Tb') \|^2;
  \end{align}
  therefore, $d(\Tb, \Tb' \mid \wb, \Lambdab)$ is a valid distance function because it equals $\| \overline \Psib_{\wb, \Rb} (\Tb) -  \overline \Psib_{\wb, \Rb} (\Tb') \|$, the Euclidean distance between $\overline \Psib_{\wb, \Rb} (\Tb)$ and $\overline \Psib_{\wb, \Rb} (\Tb')$. 

  \underline{\textbf{Proof of (2):}} 
  Using \eqref{eq:pf-athm1-5}, we have that $\overline \Psib_{\wb, \Rb} (\Tb)$ is the feature map of $\kappa$. 
  
  \underline{\textbf{Proof of (3):}} 
  Proposition 3.24 on Page 76 in \citet{ShaCri04} implies that if $\kappa(\Tb, \Tb)$ is a kernel over $(2^{\Tcal})^C \times (2^{\Tcal})^C$, then $\left\{1 + \kappa(\Tb, \Tb' \mid \wb, \Lambdab) \right\}^s$, where $s \in \NN$, is also a valid kernel over $(2^{\Tcal})^C \times (2^{\Tcal})^C$ because it is a degree-$s$ polynomial in $\kappa(\Tb, \Tb \mid \wb, \Lambdab)$ with positive coefficients. Furthermore, scaling by $\sigma^2 > 0$ implies that $\kappa_s(\Tb, \Tb' \mid \sigma^2, \wb, \Lambdab) =  \sigma^2\left\{1 + \kappa(\Tb, \Tb' \mid \wb, \Lambdab) \right\}^s$ in \eqref{aeq:thm-kern} is also a valid kernel over $(2^{\Tcal})^C \times (2^{\Tcal})^C$ for any $s \in \NN$. 

  Finally, $\kappa_\gamma^e(\Tb, \Tb'  \mid \sigma^2, \phi, \wb, \Lambdab)$ in \eqref{aeq:thm-kern} is the $\gamma$-exponential kernel defined over $[0, 1]^{(C|t^*||\Bcal|)} \times [0, 1]^{(C|t^*||\Bcal|)}$ because $d(\Tb, \Tb'  \mid \wb, \Lambdab) = \| \overline \Psib(\Tb) - \overline \Psib(\Tb)\| $ and $\overline \Psib(\Tb) \in [0, 1]^{(C|t^*||\Bcal|)}$ using its definition in \eqref{eq:pf-athm1-51}. Section 4.2.1 in \citet{RasWil06} shows that this is a valid kernel function, so $\kappa_\gamma^e(\Tb, \Tb'  \mid \sigma^2, \phi, \wb, \Lambdab)$ is a valid kernel function on $(2^{\Tcal})^C \times (2^{\Tcal})^C$; see Equation 4.18 on Page 86 in \citet{RasWil06}.   
\end{proof}

\section{Proof of Theorem 2}
\label{asec:proof-theorem-2}

Recall that the population regression and classification models are
\begin{align}
  \label{aeq:th-mdl}
  y = \mu_0(\Psi) + \epsilon, \; \epsilon \sim N(0, \tau^2_0), \; \tau^2_0 > 0, \quad \log \, \frac{\Pr (y = 1)}{\Pr (y = 0)} = \mu_0(\Psi),  \quad \Psi \in [0, 1]^{d}, 
\end{align}
where $\tau_0^2$ is the error variance, $\Tb \in (2^{\Tcal})^C$ is embedded in $[0, 1]^d$ as $\Psi = \overline \Psib_{\wb, \Rb}(\Tb)$, $d = C|t^*||\Bcal|$ using the definition of  $\overline \Psib_{\wb, \Rb}(\Tb)$ in  \eqref{eq:pf-athm1-51}, and the effect of $\Tb$ on $y$ is modeled through the feature map $ \overline \Psib_{\wb, \Rb}(\Tb)$, which is fixed for a given $\wb, \Rb$. The distance between $f(\cdot)$ and $\mu_0(\cdot)$ is defined as the empirical norm
\begin{align}
  \label{aeq:th-met}
  \| f - \mu_0 \|_n &=  \left( \frac{1}{n} \sum_{i=1}^n \left[ f \left\{ \overline \Psib_{\wb, \Rb}(\Tb_i) \right\} - \mu_0\left\{ \overline \Psib_{\wb, \Rb}(\Tb_i) \right\}  \right]^2 \right)^{1/2}  \nonumber\\
  &\equiv \left[ \frac{1}{n} \sum_{i=1}^n \{f(\Psi_i) - \mu_0(\Psi_i)\}^2  \right]^{1/2} , 
\end{align}
where $\| \cdot \|_n$ conditions on the covariates $\Tb_1, \ldots, \Tb_n$ and its dependence on $\wb$ and $\Rb$ is suppressed. We now restate Theorem 2.
\begin{theorem}\label{athm2}
  Assume that the regression and classification models in \eqref{aeq:th-mdl} are true, the Mat\'ern covariance kernel has smoothness $\nu_1$, and $\sigma$ and $\phi$ parameters of the covariance function are known. 
  \begin{enumerate}
  \item If $\mu_0$ is $\nu_2$-regular in the sense that  $\mu_0 \in \mathfrak{C}^{\nu_2}[0, 1]^d \cap \mathfrak{H}^{\nu_2}[0, 1]^d$ with $\nu_2 \leq \nu_1$, then as $n \rightarrow \infty$,
    \begin{align}
      \label{aeq:mat-rate}
      &\Pi_{\text{Mat}, n}^{\text{Reg}} \left\{ f : \| f - \mu_0 \|_{n}  > C_1 \epsilon_n \mid \yb \right\} \rightarrow 0 \text{ in } \PP_0^n \text{-probability}, \nonumber \\
      &\Pi_{\text{Mat}, n}^{\text{Clas}} \left\{ f : \| f - \mu_0 \|_{n} > C_2 \epsilon_n \mid \yb \right\} \rightarrow 0 \text{ in } \PP_0^n \text{-probability},
    \end{align}    
    are satisfied for every $\wb, \Rb$ with $\epsilon_n = n^{- \nu_2 / (2 \nu_1 + d)}$, where $C_1, C_2$ are positive constants.
  \item If $\mu_0$ is infinitely smooth functions in the sense that  $\mu_0 \in \mathfrak{A}^{r,\gamma}[0, 1]^d$ for some $r \geq 1$ and $\gamma > 0$, then as $n \rightarrow \infty$
    \begin{align}
      \label{aeq:SE-rate}
      &\Pi_{\text{SE}, n}^{\text{Reg}}  \left\{ f : \| f - \mu_0 \|_{n}  > C_1 \epsilon_n \mid \yb \right\} \rightarrow 0 \text{ in } \PP^n_0 \text{-probability}, \nonumber \\
      &\Pi_{\text{SE}, n}^{\text{Clas}} \left\{ f : \| f - \mu_0 \|_{n} > C_2 \epsilon_n \mid \yb \right\} \rightarrow 0 \text{ in } \PP_0^n \text{-probability},
    \end{align}
    are satisfied  for every $\wb, \Rb$  with $\epsilon_n = n^{-1/2}(\log n)^{\max (1/r, 1/2 + d /4 )} $,  where $C_1, C_2$ are positive constants.
  \end{enumerate}     
\end{theorem}

\begin{proof}
  We note that $\Psi$ indexes the GP prior and $\Psi \in [0, 1]^{C|t^*||\Bcal|}$, where $C, |t^*|, |\Bcal|$ are defined in Theorem \ref{athm1}; therefore, we can directly apply Theorem 11.22, Theorem 11.23, and the results in Section 11.4.4 of \citet{GhoVan17} for Matern and SE GP priors indexed by a covariate lying in $ [0,1]^d$ with $d = C|t^*||\Bcal|$. Specifically, Lemma 11.36 and Lemma 11.37 imply that \eqref{aeq:mat-rate} is satisfied with $\epsilon_n = n^{- \nu_2 / (2 \nu_1 + d)}$ for a given $\wb, \Rb$.  Similarly, Lemma 11.38 and Lemma 11.41 imply that \eqref{aeq:SE-rate} is satisfied with $\epsilon_n = n^{-1/2}(\log n)^{\max (1/r, 1/2 + d /4 )}$. 
\end{proof}

\section{Sampling Algorithm for Posterior Inference}
\label{asec:sampl-algor-post}

\subsection{GP Regression}
\label{asec:gp-regression}

Recall Equation (11) in the main manuscript.
Denote the kernel matrix as $\Kb_{\thetab}$ with $\{\Kb_{\thetab}\}_{ij} = \kappa_{f}(\Tb_i, \Tb_j)$, where $f \in \{\text{exp}, \text{SE}\}$ and $i, j \in \{1, \ldots, n\}$. Let $\yb = (y_1, \ldots, y_n)^\top$, $\Xb$ be the $n \times p$ matrix with $\xb_i^\T$ as its $i$th row, $\fb = \{f(\Tb_1), \ldots, f(\Tb_n)\}^\T$, and $\epsilonb = (\epsilon_1, \ldots, \epsilon_n)^\T$. Then, the hierarchical Bayesian model for $\yb$ is
\begin{align}
  \label{aeq:der1}
  \yb = \Xb \betab + \fb + \epsilonb, \quad \epsilonb \mid \sigma^2, \alpha \sim N(\zero, \sigma^2 \alpha \Ib), \quad  (\betab, \sigma^2) \propto \sigma^{-2}, \quad \fb \mid \thetab \sim N(\zero, \Kb_{\thetab}),
\end{align}
where $\betab$ and $\fb$ are assumed to be independent apriori, $\alpha = \tau^2/\sigma^2$ is the inverse of signal-to-noise ratio, and $\phi \in (a_{\phi}, b_{\phi})$, where $0 < a_{\phi} < b_{\phi}$.  We assign a prior on $(\phi, \alpha)$ through $(u_1, u_2) \in \RR^2$ as 
\begin{align}
  \label{aeq:81}
  \phi = a_{\phi} + (b_{\phi} - a_{\phi}) / (1 + e^{-u_1}), \quad \alpha = e^{u_2}, \quad (u_1, u_2)^\top \sim N\{\zero, \diag(b_1, b_2)\}, 
\end{align}
where $b_1, b_2 > 0$. The parameters $\sigma^2,\phi$ are non-identified in $\kappa_{\text{exp}}, \kappa_{\text{SE}}$, but this  does not affect the inference on $f$ or prediction of the response \citep{RasWil06}.

First, we derive the full conditional of $(\sigma^2, \betab)$ and the likelihood for drawing $\phi$ and $\alpha$. Marginalizing over $\fb$ in \eqref{aeq:der1} gives
\begin{align}
  \label{aeq:der3}
  \yb = \Xb \betab + \etab, \quad \etab \sim N (\zero, \sigma^2 \Cb_{yy}), \quad \Cb_{yy} = \Kb_{\thetab} / \sigma^2 + \alpha \Ib.
\end{align}
If $\Lb$ is a lower triangular matrix such that $\Cb_{yy} = \Lb \Lb^\top$, then define
\begin{align}
  \label{aeq:82}
  \tilde \yb = \Lb^{-1} \yb, \quad \tilde \Xb = \Lb^{-1} \Xb, \quad \hat \betab = (\tilde \Xb^\T \tilde \Xb)^{-1} \tilde \Xb^\top \tilde \yb, \quad \hat{\tilde \yb} = \tilde \Xb \hat \betab, \quad
  \tilde \etab = \Lb^{-1} \etab,
\end{align}
which reduces \eqref{aeq:der3} to
\begin{align}
  \label{aeq:der4}
  \tilde \yb  = \tilde \Xb \betab + \tilde \etab, \quad N(\zero, \sigma^2 \Ib).
\end{align}
The full conditional of $(\sigma^2, \betab)$ given the rest is
\begin{align}
  \label{aeq:der5}
  \sigma^2, \betab \mid \yb, \alpha, \phi \propto (2 \pi)^{-n/2} (\sigma^2)^{-n/2 - 1} e^{-\frac{1}{2 \sigma^2} \| \tilde \yb - \hat {\tilde \yb} \|^2} e^{-\frac{1}{2 \sigma^2} (\betab - \hat \betab)^\top (\tilde \Xb \tilde \Xb) (\betab - \hat \betab)}; 
\end{align}
therefore, we generate $\sigma^2 \mid \yb, \alpha, \phi$ and $\betab \mid \sigma^2, \yb, \alpha, \phi$ as
\begin{align}
  \label{aeq:der6}
  \sigma^2  \mid \yb, \alpha, \phi \sim \frac{\| \tilde y - \hat {\tilde y} \|^2} {\chi^2_{n-p}}, \quad
  \betab \mid  \sigma^2, \yb, \alpha, \phi \sim N \{\hat \betab, \sigma^2 (\tilde \Xb^\top \tilde \Xb)^{-1}\}
\end{align}
respectively, where $\chi^2_{n-p}$ is a chi-square random variable with $n-p$ degrees of freedom. The likelihood of $\phi, \alpha$ given $\yb, \sigma^2, \betab$ from \eqref{aeq:der1} and \eqref{aeq:81}. We use it in Elliptical Slice Sampling (ESS)  for drawing ($\phi$, $\alpha$). We set $\tau^2 = \alpha \cdot \sigma^2$.

Second, we drive the form of the full conditionals of $\fb^* = \{f(\Tb^*_1), \ldots, f(\Tb^*_m)\}^\top$ and $\yb^* = \{y^*_1, \ldots, y^*_m\}^\top$. For notational convenience, define $\Kb_*$ and $\Kb_{**}$ to be the matrices satisfying $\{\Kb_*\}_{ij} = \kappa_f(\Tb_i, \Tb^*_j)/ \sigma^2$ and $\{\Kb_{**}\}_{jj'} = \kappa_f(\Tb_j^*, \Tb^*_{j'})/ \sigma^2$, where $i, i' \in \{1, \ldots, n\}$ and $j, j' \in \{1, \ldots, m\}$. Then, the GP prior on $f(\cdot)$ and \eqref{aeq:der1} imply that
\begin{align}
  \label{aeq:der7}
  \begin{bmatrix}
    \yb \\
    \fb^*
  \end{bmatrix} \sim
  N \left(
  \begin{bmatrix}
    \Xb \betab \\
    \zero
  \end{bmatrix},
  \begin{bmatrix}
    \sigma^2 \Cb_{yy} & \sigma^2 \Kb_*\\
    \sigma^2 \Kb_*^\top & \sigma^2 \Kb_{**}
  \end{bmatrix}
\right);
\end{align}
the full conditional distributions of $\fb^*$ and $\yb^*$ are
\begin{align}
  \label{aeq:der8}
  \fb^* &\mid \yb, \sigma^2, \betab, \phi, \alpha \sim N(\mb_*, \Vb_*), \quad \yb^* \mid \betab, \fb^*, \tau^2  \sim N(\Xb^* \betab + \fb^*, \tau^2 \Ib), \nonumber\\
  \mb_* &= \Kb_*^\top \Cb_{yy}^{-1} (\yb - \Xb \betab), \quad 
  \Vb_* = \sigma^2 \left( \Kb_{**} - \Kb_{*}^\top \Cb_{yy}^{-1} \Kb_{*} \right) .
\end{align}
The derivation of the sampling algorithm GP regression is complete.

\subsection{GP Classification}
\label{asec:gp-classification}

Recall Equation (17) in the main manuscript. The hierarchical Bayesian model for $\yb$ is
\begin{align}
  \label{eq:cder1}
  y_i \sim \text{Bernoulli}(p_i), \quad p_i = (1 + e^{- \psi_i})^{-1}, \quad \psi_i = \xb_i^\top \betab + f(\Tb_i), \quad i = 1, \ldots, n.
\end{align}
The choice of prior distributions on $\betab, f(\cdot)$ and the kernels remain the same as in \eqref{aeq:der1}, but the prior $\thetab = (\phi, \sigma^2)^\top$ is assigned through $(u_1, u_2) \in \RR^2$ as 
\begin{align}
  \label{aeq:81log}
  \phi = a_{\phi} + (b_{\phi} - a_{\phi}) / (1 + e^{-u_1}), \quad \sigma^2 = e^{u_2}, \quad (u_1, u_2)^\top \sim N\{\zero, \diag(b_1, b_2)\}. 
\end{align}
where $b_1, b_2 > 0$. Let $\psib = \{\psi_1, \ldots, \psi_n\}$, $\Ab = [\Xb \; \Ib]$, and $\bb =(\betab, \fb)^\top$. Then, \eqref{eq:cder1} reduces to
\begin{align}
  \label{eq:6}
  \psib = \Xb \betab + \fb = [\Xb \; \Ib] \bb  = \Ab \bb.
\end{align}
For notational convenience, define $\Kb_{\thetab}$, $\Kb_*$, and $\Kb_{**}$ to be the matrices satisfying $\{\Kb_{\theta}\}_{ii'} = \kappa_f(\Tb_i, \Tb_{i'})$, $\{\Kb_*\}_{ij} = \kappa_f(\Tb_i, \Tb^*_j)$, and $\{\Kb_{**}\}_{jj'} = \kappa_f(\Tb_j^*, \Tb^*_{j'})$, where $i, i' \in \{1, \ldots, n\}$ and $j, j' \in \{1, \ldots, m\}$.

First, we augment $(y_i, x_i, \Tb_i)$ ($i=1, \ldots, n$) with random variables $\omega_1, \ldots, \omega_n$ specific to every observation such that $\omega_i$ are marginally distributed as PG(1, 0), where PG is the P\'olya-Gamma distribution with parameters $b=1$ and $c=0$, respectively. Let $\bar{\yb} = (y_1 - 1/2, \ldots, y_n - 1/2)^{\T}$, $\omegab = (\omega_1, \ldots, \omega_n)^\T$, $n \times n$ diagonal matrix $\Omegab = \diag(\omegab)$, pseudo responses $\zb = \Omegab^{-1} \bar \yb$. Then, Equation (19) in the main manuscript implies that the conditional likelihood of $\bb$ given $\yb, \omegab$ and the associated model are defined as 
\begin{align}
  \label{aeq:log2a}
  \ell (\etab \mid \yb, \omegab) \propto e^{ - \frac{1}{2} (\Ab \bb - \zb)^\top \Omega (\Ab \bb - \zb)}, \quad \zb = \Ab \bb + \epsilonb = \Xb \betab + \fb + \epsilonb, \quad \epsilonb \sim N(0, \Omegab^{-1}),
\end{align}
The full conditional of $\betab$ is obtained by first marginalizing over $\fb$ and rewriting the model for $\zb$ in \eqref{aeq:log2a} as
\begin{align}
  \label{eq:cder2}
  \zb = \Xb \betab + \xib, \quad \xib \sim N(0, \Cb_{zz}), \quad \Cb_{zz} = \Kb_{\thetab} + \Omegab^{-1} , .
\end{align}
Following \eqref{aeq:82}, if $\Lb$ is a lower triangular matrix such that $\Cb_{zz} = \Lb \Lb^\top$, then define
\begin{align}
  \label{aeq:821}
  \tilde \zb = \Lb^{-1} \zb, \quad \tilde \Xb = \Lb^{-1} \Xb, \quad \hat \betab = (\tilde \Xb^\T \tilde \Xb)^{-1} \tilde \Xb^\top \tilde \zb. 
\end{align}
Because we have imposed a flat prior on $\betab$, we use the same arguments used in \eqref{aeq:der5} to obtain that  
\begin{align}
  \label{eq:cder3}
  \betab \mid \yb, \omegab, \phi, \sigma^2 \propto e^{-\frac{1}{2} \| \tilde \zb - \hat {\tilde \zb} \|^2} e^{-\frac{1}{2} (\betab - \hat \betab)^\top (\tilde \Xb \tilde \Xb) (\betab - \hat \betab)}, \quad \betab \mid \yb, \omegab, \phi, \sigma^2 \sim N\{\hat \betab, (\tilde \Xb^\T \tilde \Xb)^{-1}\}.
\end{align}
The likelihood of $\phi, \sigma^2$ given $\yb, \betab$ from \eqref{eq:cder2} and \eqref{aeq:81log}. We use ESS  for drawing ($\phi$, $\sigma^2$). 

Second, we derive the form of the full conditionals of $\fb^* = \{f(\Tb^*_1), \ldots, f(\Tb^*_m)\}^\top$ and $\yb^* = \{y^*_1, \ldots, y^*_m\}^\top$. The GP prior on $f(\cdot)$ and \eqref{aeq:log2a} imply that
\begin{align}
  \label{eq:cder7}
  \begin{bmatrix}
    \zb \\
    \fb^*
  \end{bmatrix} \sim
  N \left(
  \begin{bmatrix}
    \Xb \betab \\
    \zero
  \end{bmatrix},
  \begin{bmatrix}
    \Cb_{zz} & \Kb_*\\
    \Kb_*^\top & \Kb_{**}
  \end{bmatrix}
\right);
\end{align}
the full conditional distributions of $\fb^*$ and $\yb^*$ are
\begin{align}
  \label{aeq:der81}
  \fb^* &\mid \yb, \betab, \phi, \sigma^2 \sim N(\mb_*, \Vb_*), \quad y^*_i \mid \betab, \fb^* \sim \text{Bernoulli}(p^*_i), \\
  \mb_* &= \Kb_*^\top \Cb_{zz}^{-1} (\zb - \Xb \betab), \quad 
  \Vb_* = \Kb_{**} - \Kb_{*}^\top \Cb_{zz}^{-1} \Kb_{*} , \quad
          p^*_i = (1 + e^{- \xb_i^{* \top} \betab - f(\Tb^*_i)})^{-1}.\nonumber
\end{align}
Finally, Theorem 1  in \citet{Poletal13} implies that the full conditional of  $\omega_i$ given $\Xb, \betab, \fb$ is PG$\{1, |\xb_i^\top \betab + (\fb)_i|\}$ for $i = 1, \ldots, n$. The derivation of the sampling algorithm GP classification is complete.

\section{Simulations for GP Regression}
\label{asec:simul-regr}

We now summarize the simulation results for GP regression with SE kernel that are based on Section 5.2 of the main paper.  The three regression models are based on the three classification models in Equations (27), (28), and (29)  of the main paper. We use the same methods that are used for evaluating the performance of GP classification in the main manuscript. The performance of all the methods is compared using the mean square error (MSE), mean square prediction error (MSPE), and the predictive coverage based on a 95\% confidence or credible interval.

\emph{\textul{First two simulations.}} Recall that in first two simulations, there are four chronic conditions. The first chronic condition is denoted using the codes in $\{A, B, AA, BB, BA, AB\}$, where every code denotes a ``diagnosis''. Replacing $(A, B)$  with $(C, D)$, $(E, F)$, and $(G, H)$ in this set gives the set of codes for the second, third, and fourth chronic conditions, respectively. The chronic conditions are further structured into two groups: the first group includes the first and third chronic conditions and the second group includes the remaining two. Our simulation assigns nine codes to every patient, where six out of the nine codes are from the 12 codes defining first or second group of chronic conditions and the remaining three come from the other chronic condition group. 

Our regression models add error terms to the linear predictors in the classification models of the main paper. Let $z_{ij}^{(c)}$ be a dummy variable indicating the presence of the $j$th code in the $c$th chronic condition in the $i$th patient ($c=1, \ldots, 4$; $j=1, \ldots, 6$; $i = 1, \ldots, n+m$), where codes follow the dictionary order, $n$ is the training data size, and $m$ testing data size. The first and second group dummy variables for the $i$th patient are $\{z_{ij}^{(c)}: j=1, \ldots, 6; c = 1, 3\}$ and $\{z_{ij}^{(c)}: j=1, \ldots, 6; c = 2, 4\}$. The response $y_i$ is simulated as
\begin{align}
  \label{eq:sup-sim1}
  y_i = 0.1 +  0.2 x_i + \sum_{c=1}^4   \sum_{j=1}^6 (-1)^c  z_{ij}^{(c)} + \epsilon_i, \quad x_i \overset{\text{ind.}}{\sim} N(0,1),  \quad \epsilon_i \overset{\text{ind.}}{\sim} N(0,\tau^2_0)
\end{align}
for $i = 1, \ldots, n+m$. The second simulation study is a slight modification of the first. Dropping the $\xb_i^\top \betab$ term in \eqref{eq:sup-sim1}, we modify it to only include the interaction of $\{AB, BA\}$, $\{CD, DC\}$, $\{EF, FE\}$, and $\{GH, HG\}$, respectively, as
\begin{align}
  \label{eq:sup-sim2}
  y_i &= 2 \left( z_{i \text{EF}}^{(2)} \cdot z_{i \text{FE}}^{(2)}  + z_{i \text{GH}}^{(4)} \cdot z_{i \text{HG}}^{(4)} \right) - \left( z_{i \text{AB}}^{(1)} \cdot z_{i \text{BA}}^{(1)} + z_{i \text{CD}}^{(3)} \cdot z_{i \text{DC}}^{(3)} \right) 
  + \epsilon_i, \nonumber\\
  x_i &\overset{\text{ind.}}{\sim} N(0,1),  \quad \epsilon_i \overset{\text{ind.}}{\sim} N(0,\tau^2_0),
\end{align}
for $i = 1, \ldots, n+m$. Setting $n=1000$, $m=100$, and $\tau_0^2 = 0.01$ in both simulations, we simulate the data and replicate the two simulation setups 10 times. 

In both simulations, our method performs the best in terms of coverage and is among the top performers in terms of MSE and MSPE (Table \ref{atab:sim1}). It achieves the nominal predictive CI coverage in both simulations. Our method overestimates $\tau^2$ slightly in both simulations resulting in a small MSE. Compared to support vector machine (SVM), kernel ridge regression (KRR), random forest, and linear regression and its penalized extensions, the MSE and MSPE of the GP with SE kernel are slightly larger in the first simulation and significantly smaller in the second simulation, respectively. This result is expected because the first simulation favors linear regression, whereas the second includes interactions of dummy variables which are only captured by SVM, KRR, and the GP with SE kernel.   

The SE kernel matrix is also useful in improving the performance of linear regression and its penalized extensions. If we use features obtained from the SE kernel matrix as predictors in linear regression and its penalized extensions and random forest, then this leads to substantial reduction of MSE and MSPE in the second simulation over the same methods that only use dummy variables as the predictors. Furthermore, using features obtained from the SE kernel yields MSE and MSPE that are significantly smaller than those for SVM and KRR in the second simulation.

\emph{\textul{Third simulation.}}
For this simulation, we used ICD codes of patients in the UIHC EHR data with brain and other nervous system or breast as the cancer primary sites. For the $i$th (pseudo) patient in this subset, let 
$z_i = \| \Psib (\Tb_i) \|^2$ and $\delta_i=1$ if the patient with $\Tb_i$ code has brain and other nervous system cancer and $\delta_i=-1$ otherwise, where $\Psib (\Tb_i)$ is defined in \eqref{aeq:3}. The response $y_i$ is simulated independently as 
\begin{align}
  \label{eq:sip-sim3}
  y_i = 0.1 +  0.2 x_i + 3\tan(z_i) + 3\delta_i, \quad x_i \overset{\text{ind.}}{\sim} N(0,1),  \quad i = 1, \ldots, n+m,
\end{align}
where $n=1000$, $m=100$, and $\tau_0^2 = 0.01$. We replicate this simulation setup 10 times. 

In this simulations, our method performs the best in terms of coverage and is the top performers in terms of MSE and MSPE (Table \ref{atab:sim1}). It achieves the nominal predictive CI coverage and overestimates $\tau^2$ slightly but has a small MSE. Compared to SVM, KRR, random forest, linear regression, and its penalized extensions, the MSE and MSPE of the GP with SE kernel are significantly smaller. This result is expected because this simulation includes nonlinear periodic trend which are only captured by SVM, KRR, and the GP with SE kernel.   

The SE kernel matrix is also useful in improving the performance of linear regression and its penalized extensions in this simulation. If we use features obtained from the SE kernel matrix as predictors in linear regression and its penalized extensions and random forest, then this leads to substantial reduction of MSE and MSPE over the same methods that only use dummy variables as the predictors. Furthermore, using features obtained from the SE kernel yields MSE and MSPE that are significantly smaller than those for SVM and KRR. 

Similar to the GP classification simulation results in the main manuscript, the distinguishing feature of our method is that the posterior draws of parameters and response can be used for quantifying uncertainty in predictions and parameter estimates, which is key in biomedical applications. We conclude that the GP with SE covariance kernel is better than other kernel based methods in quantifying uncertainty and in modeling the effect of chronic conditions on the response. 

\begin{table}[ht]
  \caption{Results of the three regression simulation studies. Every entry in the table is the average of its values across ten replications. We use `-' if the R package did not provide an output for computing the metric. }
  \label{atab:sim1}
  \centering
{\tiny
  \begin{tabular}{|r|c|c|c|c|}
    \hline
    \multicolumn{5}{|c|}{\textbf{FIRST SIMULATION}} \\
    \hline    
    & MSE ($\betab$) & MSE ($\tau^2$) & MSPE & Predictive Coverage \\ 
    \hline
    SE-GP & 0.40 & 0.00 & 0.13 & 0.95 \\ 
    \hline
    & \multicolumn{4}{c|}{Linear Regression with SE Kernel-Based Covariates} \\
    \hline 
    No Penalty & 0.07 & 0.32 & 0.57 & 0.15 \\ 
    Ridge Penalty & 0.39 & 0.36 & 0.61 & - \\ 
    Lasso Penalty & 0.07 & 0.32 & 0.57 & - \\ 
    \hline    
    & \multicolumn{4}{c|}{Linear Regression} \\
    \hline 
    No Penalty & 81.07 & 0.00 & 0.01 &  0.12 \\ 
    Ridge Penalty & 0.00 & 0.00 & 0.01 & - \\ 
    Lasso Penalty & 25.64 & 0.00 & 0.01 &-  \\ 
    \hline  
    & \multicolumn{4}{c|}{Kernel Ridge Regression}  \\
    \hline 
    Spectrum & - & 0.08 & 0.09 &-  \\ 
    Boundrange & - & 0.01 & 0.01 & - \\ 
    \hline
    & \multicolumn{4}{c|}{Support Vector Machine}  \\
    \hline 
    Spectrum & -  & 0.10 & 0.10  &-  \\ 
    Boundrange & - & 0.02 & 0.02 & - \\
    \hline
    & \multicolumn{4}{c|}{Random Forest}  \\
    \hline 
    Default & -  & 130.07 & 11.38 &-  \\ 
    SE Kernel-Based Covariates & - & 0.01 & 0.13 & - \\ 
    \hline
    \multicolumn{5}{c|}{\textbf{SECOND SIMULATION}} \\
    \hline
    & MSE ($\betab$) & MSE ($\tau^2$) & MSPE & Predictive Coverage \\ 
    \hline
    SE-GP & 1.10 & 0.00 & 0.15 & 0.95 \\ 
    \hline
    & \multicolumn{4}{|c|}{Linear Regression with SE Kernel-Based Covariates} \\
    \hline 
    No Penalty & 5.00 & 0.67 & 0.83 &  0.21 \\ 
    Ridge Penalty & 0.62 & 0.67 & 0.83 & - \\ 
    Lasso Penalty & 0.49 & 0.67 & 0.83 & - \\ 
    \hline    
    & \multicolumn{4}{c|}{Linear Regression} \\
    \hline 
    No Penalty & 6.97 & 0.43 & 0.67 & 0.04 \\  
    Ridge Penalty & 0.25 & 0.45 & 0.67 & - \\ 
    Lasso Penalty & 0.28 & 0.45 & 0.67 & - \\  
    \hline  
    & \multicolumn{4}{c|}{Kernel Ridge Regression}  \\
    \hline 
    Spectrum & - & 0.41 & 0.42 & -  \\ 
    Boundrange & - & 0.49 & 0.50 & - \\
    \hline
    & \multicolumn{4}{c|}{Support Vector Machine}  \\
    \hline 
    Spectrum & -  &  20.38 & 20.47  &-  \\ 
    Boundrange & - & 21.52 & 21.61 & - \\
    \hline
    & \multicolumn{4}{c|}{Random Forest}  \\
    \hline 
    Default & -  & 5.32 & 2.31  &-  \\ 
    SE Kernel-Based Covariates & - & 0.08 & 0.29 & - \\ 
    \hline
    \multicolumn{5}{|c|}{\textbf{THIRD SIMULATION}} \\
    \hline
    & MSE ($\betab$) & MSE ($\tau^2$) & MSPE & Predictive Coverage \\ 
    \hline
    SE-GP & 0.32 & 1.12 & 2.43 &  0.97 \\  
    \hline
    & \multicolumn{4}{c|}{Linear Regression with SE Kernel-Based Covariates} \\
    \hline 
    No Penalty & 5764.93 & 219.66 & 14.54 & 0.36 \\  
    Ridge Penalty & 0.73 & 219.84 & 14.54 & - \\  
    Lasso Penalty & 0.73 & 219.87 & 14.54 & - \\ 
    \hline    
    & \multicolumn{4}{c|}{Linear Regression} \\
    \hline 
    No Penalty & 84.73 & 18.62 & 1077860706.04 & 0.78 \\
    Ridge Penalty & 4.27 & 10928.80 & 90.24 & -  \\ 
    Lasso Penalty & 1.22 & 23349.05 & 109.79 & - \\ 
    \hline  
    & \multicolumn{4}{c|}{Kernel Ridge Regression}  \\
    \hline 
    Spectrum & - & 22095.18 & 22098.00 & - \\ 
    Boundrange & - & 28392.57 & 28395.72 & -\\ 
    \hline
    & \multicolumn{4}{c|}{Support Vector Machine}  \\
    \hline 
    Spectrum & -  & 51339.34 & 51343.55 & - \\ 
    Boundrange & - & 53262.98 & 53267.27 & - \\
    \hline     
    & \multicolumn{4}{c|}{Random Forest}  \\
    \hline 
    Default & -  & 30745.21 & 30748.61 & - \\ 
    SE Kernel-Based Covariates & - & 46.13 & 4.29 & - \\ 
    \hline        
  \end{tabular}
}%
\end{table}

\end{document}